%% file: main.tex
\documentclass[sigconf,table]{acmart}

\usepackage{subcaption}
\usepackage{paralist}
\usepackage{listings}
\usepackage[skip=0.2pt]{caption}
\usepackage{url}
\usepackage[vlined,ruled]{algorithm2e}
\usepackage{tikz}
\usepackage{hhline}
\usepackage{multirow}
\usetikzlibrary{calc,angles,intersections,quotes,arrows,automata,fit,shapes,trees,topaths,calc,shapes.geometric,arrows.meta,positioning,positioning,chains,fit,shapes,calc,matrix}

\newcommand{\paratitle}[1]{\noindent{\bf #1}}

\colorlet{myPurp}{purple!40!gray!60}

\makeatletter
\newcommand*{\rom}[1]{\expandafter\@slowromancap\romannumeral #1@}
\makeatother

\SetCommentSty{mycommfont}

\newcommand{\reva}[1]{\leavevmode\color{black}{#1}}
\newcommand{\revb}[1]{\leavevmode\color{black}{#1}}
\newcommand{\revc}[1]{\leavevmode\color{black}{#1}}
\newcommand{\common}[1]{\leavevmode\color{black}{#1}}

\newtheorem{theorem}{Theorem}[section]
\newtheorem{proposition}[theorem]{Proposition}
\newtheorem{example}[theorem]{Example}

\newtheorem{definition}[theorem]{Definition}

\newtheorem{lemma}[theorem]{Lemma}

\newcommand{\prob}{C-Extension}

\newcommand{\noJoin}{invalid}

\newcommand{\joinView}{\ensuremath{V_{Join}}}

\renewcommand\footnotetextcopyrightpermission[1]{} 

\begin{document}
\fancyhead{}

%
\title{Synthesizing Linked Data Under \texorpdfstring{\\} Cardinality and Integrity Constraints}

%
\settopmatter{authorsperrow=3}
\author{Amir Gilad}
\authornote{Both authors contributed equally to this research.}
\affiliation{%
\institution{Duke University}
\country{}
}
\email{agilad@cs.duke.edu}

\author{Shweta Patwa}
\authornotemark[1]
\affiliation{%
  \institution{Duke University}
  \country{}
}
\email{sjpatwa@cs.duke.edu}

\author{Ashwin Machanavajjhala}
\affiliation{%
  \institution{Duke University}
  \country{}
}
\email{ashwin@cs.duke.edu}

%
%
\begin{abstract}
The generation of synthetic data is useful in multiple aspects, from testing applications to benchmarking to privacy preservation. Generating the \textit{links} between relations, subject to \textit{cardinality constraints} (CCs) and \textit{integrity constraints} (ICs) is an important aspect of this problem. 
Given instances of two relations, where one has a foreign key dependence on the other and is missing its foreign key ($FK$) values, and two types of constraints: (1) CCs that apply to the join view and (2) ICs that apply to the table with missing $FK$ values, our goal is to impute the missing $FK$ values such that the constraints are satisfied. 
We provide a novel framework for the problem based on declarative CCs and ICs. 
We further show that the problem is NP-hard and propose a novel two-phase solution that guarantees the satisfaction of the ICs. 
Phase \rom{1} yields an intermediate solution accounting for the CCs alone, and relies on a hybrid approach based on CC types. For one type, the problem is modeled as an Integer Linear Program. For the others, we describe an efficient and accurate solution. We then combine the two solutions. 
Phase \rom{2} augments this solution by incorporating the ICs and uses a coloring of the conflict hypergraph to infer the values of the $FK$ column. 
Our extensive experimental study shows that our solution scales well when the data and number of constraints increases. We further show that our solution maintains low error rates for the CCs. 
\end{abstract}

\maketitle
%




\input{intro.tex}
\vspace{-1mm}
\input{prelim.tex}
\input{model.tex}
\vspace{-1mm}
\input{overview.tex}
\input{cc_sol}
\vspace{-1mm}
\input{dc_sol.tex}
\vspace{-1mm}
\input{experiments.tex}
\vspace{-1mm}
\input{related.tex}
\vspace{-1mm}
\input{conclusions.tex}

\paratitle{Acknowledgements.} This work was supported by the National Science Foundation under grant 1408982; and by DARPA and SPAWAR under contract N66001-15-C-4067.

\clearpage

\bibliographystyle{plain}
\bibliography{bibtex.bib}

\clearpage
\input{revision_appendix}

\end{document}

%% file: intro.tex
\section{Introduction}\label{sec:intro}
In recent years, we have witnessed an increase in data-centric applications that call for efficient testing over reliable databases with certain desired qualities \cite{BrunoC05,LoCH10}. 
Existing benchmarks such as TPC-H \cite{tpchPaper,tpch} may not possess the desired characteristics for testing a specific application as they may not have the needed statistical qualities or the correct Integrity Constraints (ICs).
The field of {\em data generation} \cite{MannilaR89,GraySEBW94,HoukjaerTW06,BinnigKLO07,Arasu2011,RablDFSJ15,FazekasK18,SanghiSHT18} has proven effective in this respect. 
Two prominent challenges in this field are: (1) the generation
of links between different tables, i.e., aligning foreign keys with
primary keys based on Cardinality Constraints (CCs) \cite{Arasu2011}, and
(2) ensuring that the data will satisfy a set of expected ICs \cite{SoltanaSB17}. 

In particular, when the real data is sensitive and access to it is heavily regulated, users often need to wait months or years to get access to the real data before they can even start writing data analysis programs. One solution is to generate realistic synthetic data that satisfies some CCs and ICs so that users can: (a) start writing code to analyse the data, (b) test it locally, and (c) evaluate whether access to the data would be useful for their purposes even before they get access to the real data. However, current methods for generating synthetic data under privacy constraints (especially state-of-the-art standards like differential privacy \cite{Dwork06}) do not handle data with a combination of CCs or statistical constraints and ICs. Most, (e.g., \cite{ZhangCPSX17,HeMD14,SnokeS18}), only handle statistical constraints. 

Furthermore, there has been a lot of recent work on answering count queries under differential privacy (e.g., Matrix mechanism \cite{LiMHMR15}, HDMM  \cite{McKennaMHM18}) and in particular over relational databases \cite{KotsogiannisTHF19}. A key challenge when answering queries especially over relational databases is that of \textit{consistency} -- are the answers outputted by a differentially private algorithm consistent with some underlying database? While there is work on using inference to enforce consistency when all the count queries are over a single view of the underlying database \cite{HayRMS10}, these techniques do not extend to the case when: (a) the underlying database is relational and query answers are over several joined views of the relations, and (b) when the underlying database needs to satisfy some ICs. One solution to this problem is to find a database that is consistent with the query answers and the ICs, and answer queries from it. While techniques for finding such a consistent database are known for single tables without ICs \cite{HayRMS10,LiHMW14,BarakCDKMT07}, no such techniques are known when there are multiple tables in a relational database with ICs. 

\revb{
Moreover, DBMS testing and other applications may require databases that conform to both CCs and ICs to make them more realistic \cite{Arasu2011,SoltanaSB17}. For instance, consider a table with the attributes $A$ and $B$. A query grouping over attributes $A$ and $B$ could return as many tuples as the cross product of the active domains of $A$ and $B$. However, if there is a Functional Dependency $A \rightarrow B$, then the output size of the group-by query is only the maximum of the active domains of the two attributes. Thus, the presence of ICs can significantly impact the performance characteristics of queries. 
}

\begin{figure}[ht]
    \centering
	\begin{footnotesize}
		\begin{minipage}{.65\linewidth}
		    \centering
            \caption*{Persons (rel. $R_1$)}\label{tbl:presons}
    		\begin{tabular}{| c | c | c | c || >{\columncolor[RGB]{204, 204, 255}}c |}
    			\hline $p_{id}$ & $Age$ & $Rel$ & \reva{$Multi$-$ling$} & $h_{id}$ \\
    			\hline $1$ & $75$ & $Owner$ & $0$ & ?\\
    			\hline $2$ & $75$ & $Owner$ & $1$ & ?\\
    			\hline $3$ & $25$ & $Owner$ & $0$ & ?\\
    			\hline $4$ & $25$ & $Owner$ & $1$ & ?\\
    			\hline $5$ & $24$ & $Spouse$ & $0$ & ?\\
    			\hline $6$ & $10$ & $Child$ & $1$ & ?\\
    			\hline $7$ & $10$ & $Child$ & $1$ & ?\\
    			\hline $8$ & $30$ & $Owner$ & $0$ & ?\\
    			\hline $9$ & $30$ & $Owner$ & $1$ & ?\\
    			\hline
    		\end{tabular}
		\end{minipage}%
        \begin{minipage}{.35\linewidth}
            \centering
            \caption*{Housing (rel. $R_2$)}\label{tbl:house}
    		\begin{tabular}{| c | c | c |}
    		     \hline $h_{id}$ & $Area$\\
    		     \hline $1$ & $Chicago$\\
    		     \hline $2$ & $Chicago$\\
    		     \hline $3$ & $Chicago$\\
    		     \hline $4$ & $Chicago$\\
    		     \hline $5$ & $NYC$\\
    		     \hline $6$ & $NYC$\\
    		     \hline
    		\end{tabular}
		\end{minipage}
	\end{footnotesize}
	\caption{Database $\mathcal{D}$ with FK $h_{id}$ missing from $R_1$}
	\label{fig:db_table}
	\vspace{-3mm}
\end{figure}

\begin{figure}[ht]
		\centering
		\begin{subfigure}{\linewidth}
		\begin{footnotesize}
        \begin{lstlisting}[mathescape=true, basicstyle=\linespread{1.5}]
$DC_{O,O}:$ $\forall t_1, t_2.~ \neg (t_1.Rel = t_2.Rel = Owner \land t_1.h_{id} = t_2.h_{id})$
$DC_{O,S,low}:$ $\forall t_1, t_2.~ \neg(t_1.Rel= Owner \land t_2.Rel= Spouse \land$
            $t_2.Age<t_1.Age-50 \land t_1.h_{id}=t_2.h_{id})$
$DC_{O,S,up}:$ $\forall t_1, t_2.~ \neg(t_1.Rel= Owner \land t_2.Rel= Spouse \land$ 
            $t_2.Age>t_1.Age+50 \land t_1.h_{id}=t_2.h_{id})$
$DC_{O,C,low}:$ $\forall t_1, t_2.~ \neg(t_1.Rel= Owner \land t_1.\reva{Multi\text{-}ling= 1} \land t_2.Rel= $ 
            $Child \land t_2.Age<t_1.Age-50 \land t_1.h_{id}=t_2.h_{id})$
$DC_{O,C,up}:$ $\forall t_1, t_2.~ \neg(t_1.Rel= Owner \land t_1.\reva{Multi\text{-}ling= 1} \land t_2.Rel=$
            $Child \land t_2.Age>t_1.Age-12 \land t_1.h_{id}=t_2.h_{id})$
        \end{lstlisting}
        \end{footnotesize}
        \caption{\normalsize Denial Constraints: $DC_{O,O}$ enforces that no two homeowners can reside in the same home, $DC_{O,S,low}$ and $DC_{O,S,up}$ together specify the permissible age range of a spouse in any home, and $DC_{O,C,low}$ and $DC_{O,C,up}$ give the age range for a child living with a \reva{multi-lingual} homeowner}\label{fig:dcs}
		\end{subfigure}
        \begin{subfigure}{\linewidth}
        \begin{small}
        \centering
        \begin{lstlisting}[mathescape=true, basicstyle=\linespread{1.5}]
        $CC_1:$ $|\sigma_{Rel=Owner,Area=Chicago}(R_1\bowtie R_2)| = 4$
        $CC_2:$ $|\sigma_{Rel=Owner,Area=NYC}(R_1\bowtie R_2)| = 2$
        $CC_3:$ $|\sigma_{Age\leq 24,Area=Chicago}(R_1\bowtie R_2)| = 3$
        $CC_4:$ $|\sigma_{\reva{Multi\text{-}ling=1},Area=Chicago}(R_1\bowtie R_2)| = 4$
        \end{lstlisting}
        \end{small}
        \caption{\normalsize Cardinality Constraints: $CC_1$ and $CC_2$ give the number of homeowners in Chicago and NYC, resp., $CC_3$ gives the number of people younger than $25$ who live in Chicago, and $CC_4$ gives the number of \reva{multi-lingual individuals} in Chicago.}\label{fig:ccs}
		\end{subfigure}
	\caption{Set of DCs and set of CCs}\label{fig:constraints}
	\vspace{-4mm}
\end{figure}

{\em In this paper, we investigate the problem of generating the links between database tables based on a set of linear CCs and a set of ICs.}

Formally, we consider two relations, $R_1$ and $R_2$, where $R_1$ has a foreign key dependence on $R_2$ and is missing all values in its foreign key column $FK$. The goal is to impute $FK$ in $R_1$ based on the given CCs and ICs. Importantly, this problem and our solutions can be extended to \revb{relational databases} with a snowflake schema~\cite{10.1145/248603.248616}, by focusing on pairs of relations linked by foreign key joins.


\begin{example}
\label{ex:running_ex_setup}
Consider the relations in Figure~\ref{fig:db_table} based on the Census database. $R_1$ describes 
people through attributes such as age, relationship to a household (e.g. owner or spouse), \reva{whether they speak more than $1$ language} and a (missing) household id, whereas $R_2$ shows 
the area for each household. 
In addition, we are given the set of ICs and CCs in Figures \ref{fig:dcs} and \ref{fig:ccs}, respectively.
The goal is to impute 
values in the $h_{id}$ column in $R_1$ 
so that the ICs and CCs are satisfied.
\end{example}


We believe that the problem we focus on is a key building block for the general problem of synthesizing data consistent with CCs and ICs for all three use-cases mentioned above. 
In particular, we believe that one can use the wealth of existing literature to synthesize individual relations consistent with CCs without the key relationships and then use our technique to fill-in the foreign keys.



\subsection*{Our Contributions}
We model the problem, give a theoretical analysis, and provide a solution for the generation of foreign keys for existing database relations while ensuring the satisfaction of a set of ICs and reducing the error of a set of CCs. 
Next, we give our main contributions.\\

\paratitle{Model and Theoretical Results:}
We define the problem of \prob\ whose input is a relation $R_1$ with an unknown foreign key dependence on a relation $R_2$, i.e., the $FK$ column in $R_1$ is missing, and a set of CCs and ICs. 
For the CCs, we define and use linear CCs that apply to $R_1\bowtie R_2$, based on \cite{Arasu2011}. For the ICs, we define a type of Denial Constraints (DCs) \cite{ChomickiM05,ChuIP13}, called Foreign Key DCs, that applies to $R_1$ and forbids tuples from having the same $FK$ value under specified conditions. 
We then show that \prob\ is NP-hard in data complexity. 
This result leads us to a two-phase heuristic solution that still ensures the satisfaction of all DCs, while tolerating possible errors in the CC counts. 

\paratitle{Solution:}
Our solution can be split into two phases:
(1) first phase (Section \ref{sec:cc}) is designed for the completion of a view \joinView\ based on CCs, where \joinView\ represents $R_1\bowtie R_2$ and is initialized with a copy of $R_1$ (without the $FK$ column) along with an empty column per non-key column in $R_2$ (due to foreign key dependence, $|R_1| = |\joinView|$), and (2) second phase (Section \ref{sec:dc}) uses the generated view \joinView\ to complete the $FK$ column in $R_1$ so that the DCs are satisfied.

    {\bf Phase \rom{1}:} 
    We provide a novel description of CC relationships that allows for \joinView\ to be completed efficiently and precisely under specific conditions \revc{(presented in Section \ref{sec:first_phase_overview})}. We further devise algorithms for this case and the general case:
    \begin{itemize}
    \itemsep0em
        \item For the general case, we devise an algorithm that models the CCs and the tuples in \joinView\ as an Integer Linear Program (inspired by \cite{Arasu2011}). From its solution, we greedily infer the values in \joinView\ for the attributes that come from $R_2$.
        \item For the special case, we devise a novel algorithm based on relationships between the CCs. 
        We show that if the CCs have containment or disjointness relationships between them (defined in Section \ref{subsec:cc_relationships}), then we can find an exact completion of \joinView\ without any errors, provided one exists.
    \end{itemize}
    Our approach is a hybrid of these two solutions 
    that employs the first solution for the subset of CCs that does not fit the special case, and employs the second solution for the subset of CCs that does. 
    
    Another novelty in our solution exploits the fact that the \reva{all-way marginals for $R_1$, i.e., counts of tuples with different combinations of values in $R_1$'s non-key columns,} have the same counts in \joinView. Thus, we augment the input set of CCs to improve accuracy. 
    
    
    {\bf Phase \rom{2}:} 
    For the second phase, we employ the concept of a \textit{conflict hypergraph} \cite{ChuIP13} and use a novel algorithm based on hypergraph coloring. 
    We model the tuples in $R_1$ as vertices and connect by an edge every set of 
    tuples that will violate a DC if 
    assigned the same foreign key. Thus, colors represent the values that the foreign keys can take in $R_1$, and a proper coloring represents a mapping of tuples to foreign keys that does not violate any DC. Due to the previous stage that considered $R_1\bowtie R_2$, tuples in $R_1$ have a certain list of permitted colors. This version of the graph coloring problem is called {\em List Coloring} \cite{achlioptas_molloy_1997} and is known to be NP-hard. 
    To color the graph, we use a greedy coloring algorithm that considers vertices in descending order by degrees. 
    The algorithm skips vertices whose list of permitted colors is subsumed by the colors assigned to their neighbors. We ensure a proper coloring by adding the least number of new colors 
    for the skipped vertices. Adding colors beyond the permitted lists corresponds to artificially adding tuples in $R_2$.

\paratitle{Experimental Evaluation}
We have implemented our solution and performed a comprehensive set of experiments on a dataset derived from the 2010 U.S. Decennial Census \cite{sexton_abowd_schmutte_vilhuber_2017}. 
We have evaluated our solution in terms of accuracy and scalability in various scenarios, several of which were used for comparison with a baseline based on \cite{Arasu2011}. 
We further examined the runtime breakdown of our approach, presenting the runtimes of phases \rom{1} and \rom{2} in our solution. 
Our results indicate that our solution incurs relatively small error for CCs and no error for DCs (as guaranteed by our theoretical analysis). Moreover, our algorithms scale well for large data sizes, and large and complex sets of CCs and DCs. 
For increasing data scales, our approach was $17$ times faster on average across different cases than the baseline we compare to.

%% file: prelim.tex
\section{Preliminaries and Model}\label{sec:prelim}
We now define the basic concepts used throughout the paper, and the \prob\ problem. 

\paratitle{Relations in a Database:}
Let $R_1$ and $R_2$ be relations over the schema attributes $(K_1, A_1,\ldots, A_p,$ $FK)$ and $(K_2, B_1, \ldots, B_q)$, respectively. An attribute $A_j$ of $R_i$ may also be called a column and is denoted by $R_i.A_j$. 
$t \in R_i$ denotes a tuple in $R_i$ and $t.A_j$ denotes the {\em cell} of column $A_j$ in tuple $t$. 
The last column in $R_1$ ($FK$) is a foreign key column that gets its values from the key column $K_2$ in $R_2$. 
The view $\joinView = R_1 \bowtie_{FK=K_2} R_2$ denotes the join of the two relations. 
If all values of a column $A_i$ are missing, it is called a {\em missing column}. 

\begin{example}
Consider a database $\mathcal{D}$ with two relations $R_1$ and $R_2$ as shown in Figure~\ref{fig:db_table}. $R_1.h_{id}$ is a missing column. The first row in $R_1$ says that $t_1.Age$ is $75$, $t_1.Rel$ is Owner and \reva{$t_1.Multi$-$ling$} is $0$.
\end{example}


\paratitle{Foreign Key Denial Constraints:}
DCs \cite{ChomickiM05} are a general form of constraints that can be written as a negated First Order Logic statement. DCs can express several types of integrity constraints like functional dependencies and conditional functional dependencies \cite{bohannon2007conditional}.
In this paper, we restrict our attention to DCs that contain a condition of the form $t_1.FK = \ldots = t_k.FK$.

\begin{definition}[Foreign Key DC]\label{def:dc}
A Foreign Key DC on a relation $R(K_1, A_1,\ldots,$ $A_p, FK)$ is defined as the following FOL statement:
\begin{small}
\begin{lstlisting}
                $\forall t_1, t_2, \ldots, t_k.~$ $\neg (p_1 \land \ldots \land p_n)$
\end{lstlisting}
\end{small}
where $p_q = t_i.A_l  \circ t_j.A_l$ or $p_q = t_i.A_l \circ c$, for $t_i,t_j \in R$, $p\geq 2$, $\circ \in \{=,<,>,\neq\}$, $c$ and $k$ are constants, and $p_n = (t_1.FK = \ldots = t_k.FK)$. 
\end{definition}

We use the terms Foreign Key DC and DC interchangeably. 

\begin{example}\label{ex:dc}
$DC_{O,O}$ (Figure~\ref{fig:dcs}), which states that two homeowners cannot be in the same home, can be formulated as follows:
\begin{small}
\begin{lstlisting}
    $\forall t_1, t_2\in R_1.~ \neg(t_1.Rel = t_2.Rel = Owner \land t_1.h_{id} = t_2.h_{id})$
\end{lstlisting}
\end{small}
Note that the restriction to Foreign Key DCs means that all constraints are on people that are in the same household.
\end{example}



\paratitle{Linear Cardinality Constraints:}
CCs form the second class of constraints that allows for the specification of the number of tuples that should posses a certain set of attribute values, which can be expressed as a selection condition. 
As standard in previous work \cite{Arasu2011, Mckenna2019}, we restrict our attention to \textit{linear} CCs. 

\begin{definition}[Linear CC, adapted from \cite{Arasu2011}]\label{def:cc}
A linear CC over a database $\mathcal{D}$ consisting of relations $R_1(K_1, A_1,\ldots, A_p, FK)$ and $R_2(K_2, B_1, \ldots,$ $B_q)$ is defined as follows:
\begin{small}
\begin{lstlisting}
                $|\sigma_{\varphi}(R_1\bowtie_{FK=K_2} R_2)| = k$
\end{lstlisting}
\end{small}
where $\varphi$ is a Boolean selection predicate over a subset of (non-key) attributes in $\mathcal{D}$, and $k\in\mathbb{N}$. 
\end{definition}

In the rest of the paper, we only refer to conjunctive selection predicates with conjuncts of the form $A_i\circ c$, where $\circ\in \{=, <, >, \leq, \geq\}$ and $c$ is in the domain of column $A_i$, though our algorithms can be extended to conditions that contain disjunction as well.


\begin{example}\label{ex:cc}
$CC_1$ (Figure~\ref{fig:ccs}), which states that the number of homeowners ($Rel=$ Owner) living in $Area=Chicago$ must equal $4$, can be written as: $|\sigma_{Rel=Owner, Area=Chicago}~ R_1\bowtie R_2| = 4$.
\end{example}

We denote by $R\vDash \sigma$ the fact that relation $R$ meets constraint $\sigma$.

%% file: model.tex

\paratitle{Problem Definition:} We now formally define the \prob\ problem and discuss its intractability. 

\begin{definition}[\prob]\label{def:prob}
Let $R_1(K_1, A_1, \dots, A_p, FK)$ and $R_2(K_2,$ $B_1, \dots, B_q)$ be two relations, where $R_1.FK$ is a foreign key mapped from $R_2.K_2$ and is empty
. 
Let $S_{DC}$ denote the set of DCs over $R_1$ and let $S_{CC}$ denote the set of linear CCs over the foreign key join between $R_1$ and $R_2$. 
\prob\ is the problem of completing all the values in $R_1.FK$ to create $\hat{R_1}$ so that (1) $\forall \sigma\in S_{DC},~ \hat{R_1} \vDash \sigma$, (2) $\forall \sigma\in S_{CC},~ \hat{R_1}\bowtie_{FK=K_2} R_2 \vDash \sigma$. 
\end{definition}

\begin{example}\label{ex:prob}
Reconsider relations $R_1$ and $R_2$ in Figure~\ref{fig:db_table}, and DCs and CCs in Figure~\ref{fig:constraints}.
A solution $\hat{R_1}$ for the \prob\ problem as defined by these relations and constraints is shown in Figure \ref{fig:db_table_complete}.
\vspace{-2mm}
\begin{figure}[!ht]
	\centering
	\begin{footnotesize}
    \caption*{Persons (rel. $R_1$)}\label{tbl:persons_complete}
	\begin{tabular}{| c | c | c | c || >{\columncolor[RGB]{204, 204, 255}}c |}
		\hline $p_{id}$ & $Age$ & $Rel$ & \reva{$Multi$-$ling$} & $h_{id}$ \\
		\hline $1$ & $75$ & $Owner$ & $0$ & $2$\\
		\hline $2$ & $75$ & $Owner$ & $1$ & $1$\\
		\hline $3$ & $25$ & $Owner$ & $0$ & $3$\\
		\hline $4$ & $25$ & $Owner$ & $1$ & $4$\\
		\hline $5$ & $24$ & $Spouse$ & $0$ & $2$\\
		\hline $6$ & $10$ & $Child$ & $1$ & $2$\\
		\hline $7$ & $10$ & $Child$ & $1$ & $2$\\
		\hline $8$ & $30$ & $Owner$ & $0$ & $5$\\
		\hline $9$ & $30$ & $Owner$ & $1$ & $6$\\
		\hline
	\end{tabular}
    \end{footnotesize}
\caption{Relation $R_1$ from Figure~\ref{fig:db_table} with FK $h_{id}$ filled-in to satisfy DCs and CCs given in Figure~\ref{fig:constraints}}\label{fig:db_table_complete}
\end{figure}
\vspace{-3mm}
\end{example}

The decision version of \prob\ is given by the same setting as in Definition \ref{def:prob}. The output is $1$ if there exists a completion of $R_1.FK$ such that all DCs and CCs are satisfied, and $0$ otherwise. 
\begin{proposition}\label{prop:hardness}
The decision problem version of \prob\ is NP-hard in data complexity.
\end{proposition}

\revc{
\vspace{-6mm}
\begin{proof}[Proof Sketch]
We describe a reduction from NAE-3SAT to \prob. 
In the NAE-3SAT problem, we are given a 3-CNF formula $\varphi$ and asked whether there is a satisfying assignment to $\varphi$ with every clause having at least one literal with the value False. 
Given a 3-CNF formula $\varphi = C_1 \land \ldots \land C_n$, where $x_1,\ldots, x_m$ are the propositional variables in $\varphi$, construct a relation $R_1(Var,\alpha,Cls,Chosen)$, where $Chosen$ is missing all values, and $Var,\alpha,Cls$ 
columns take values:
\setdefaultleftmargin{0pt}{}{}{}{}{}
\begin{enumerate}
    \itemsep0em
    \item $(x_i, 1, C_j, ?)$ if making $x_i$ \textit{True} makes $C_j$ \textit{True}
    \item $(x_i, 0, C_j, ?)$ if making $x_i$ \textit{False} makes $C_j$ \textit{True}
\end{enumerate}
We define $S_{DC}$ to be the set with the following two DCs:
\begin{small}
\begin{enumerate}
    \itemsep0em
    \item $\forall t_1, t_2.~ \neg (t_1.Var$=$t_2.Var \land t_1.\alpha \neq t_2.\alpha \land  t_1.Chosen$=$t_2.Chosen)$ 
    \item $\forall t_1, t_2, t_3.~ \neg (t_1.Cls$=$t_2.Cls$=$t_3.Cls \land t_1.Chosen$=$t_2.Chosen$= $t_3.Chosen)$
\end{enumerate}
\end{small}
CCs are not needed in the reduction. 
The goal is to complete the missing column $Chosen$ in $R_1$.
We define $R_2$ as containing two columns: a primary key column $Chosen$, and another column $E$. 
$R_2$ contains the tuples $(0, a)$ and $(1, b)$, i.e., the domain for $Chosen$ is $\{0, 1\}$. 
Intuitively, $Chosen$ encodes the satisfying assignment for $\varphi$ by assigning values to each tuple, where $t.Chosen$=$1$ iff the assignment should be $t.Var$=$t.\alpha$.
\end{proof}
}

The full proofs are detailed in Section \ref{sec:proofs} of the appendix.

%% file: overview.tex
\section{Solution Overview}\label{sec:overview}
Our solution proceeds in two phases as seen in Figure \ref{fig:detailed}. 
In phase \rom{1}, we consider the view \joinView\ representing the join of the two relations $R_1$ and $R_2$, where $R_1$ has a foreign key dependence on $R_2$, and initialize it with (non $FK$) columns from $R_1$ and an empty column per non-key column from $R_2$. 
We infer these values based on the CCs by a hybrid approach that uses both ILP \cite{Arasu2011} and a more efficient and accurate procedure for special cases. 
In phase \rom{2}, we impute $R_1.FK$ by modeling the problem as a conflict hypergraph using the DCs, and coloring it based on the inferred values in \joinView. 

\begin{figure}[ht]
    \centering
    \includegraphics[width=0.77\linewidth]{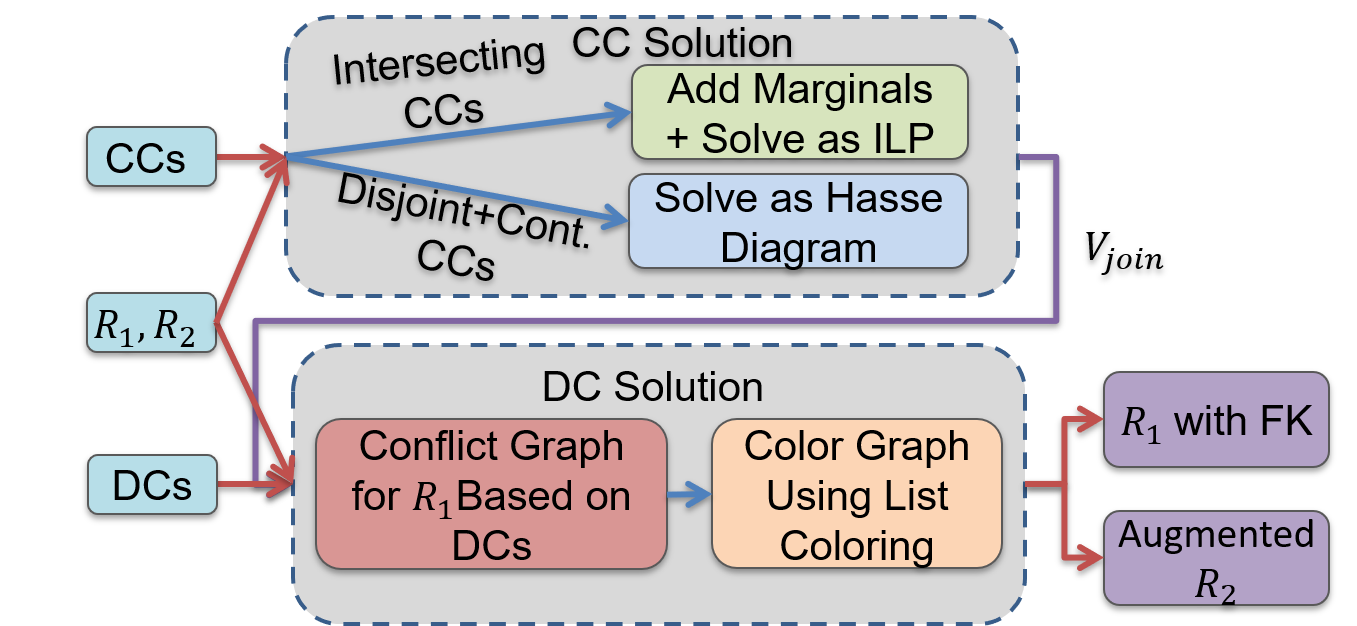}
    \caption{Solution Overview}
    \label{fig:detailed}
    \vspace{-4mm}
\end{figure}

\subsection{Overview of the First Phase}\label{sec:first_phase_overview}

Due to the foreign key dependence (Definition \ref{def:prob}), we define \joinView\ over the columns $K_1, A_1, \ldots, A_p, B_1,\ldots, B_q$ such that $t\in R_1$ implies that there is a single $t'\in\joinView$ with $t.K_1$=$t'.K_1$ and $\forall 1\leq i\leq p.~ t.A_i = t'.A_i$ 
with additional $B_1,\ldots, B_q$ entries that are initially all empty because $FK$ is missing in $R_1$. 
Therefore, $|\joinView| = |R_1|$. Our goal is to complete these columns based on the CCs.


\begin{example}\label{ex:join_view_missing}
Reconsider $R_1$ and $R_2$ shown in Figure \ref{fig:db_table} and the CCs in Figure \ref{fig:ccs}. 
The join view \joinView\ is $R_1$ as it appears in Figure \ref{fig:db_table} (without $h_{id}$) with an empty Area column (as this is the schema of $R_1\bowtie R_2$). 
Due to the foreign key dependency, we have $|\joinView| = |R_1|$, 
and \joinView\ contains a tuple for each $R_1$ tuple with the same values as in $R_1$ and an empty $Area$ value.
The reason is that the $FK$ values are missing in $R_1$. Our goal is to fill-in \joinView\ so that the CCs are satisfied.
\end{example}

We give a short description of our solution for completing \joinView. 

\paratitle{Solution as an ILP (Section \ref{subsec:cc_general}, green box in Figure \ref{fig:detailed}):}
Given a set of CCs on \joinView, we model the problem of completing the missing columns as a system of linear equations with variables accounting for counts of different tuples needed in \joinView\ to satisfy the CCs. Thus, the variables must take non-negative integer values. 
We artificially add to $S_{CC}$ \reva{all-way marginals (using the idea of \textit{intervalization} from \cite{Arasu2011} that is explained in Section \ref{sec:cc})} from $R_1$ to enhance the accuracy of the solution. \reva{For example, based on CCs given in Example~\ref{ex:running_ex_setup}, $|\sigma_{Age\leq 24, Rel=Spouse, Multi\text{-}ling=0}|=1$ gets added to $S_{CC}$.} We then assign $B_1, \ldots, B_q$ values to the tuples in \joinView\ based on the solution returned by an ILP solver.

\paratitle{Using CC Relationships for Special Cases (Section \ref{subsec:cc_relationships}, blue box in Figure \ref{fig:detailed}):}
We give a novel description of the relationships between CCs based on their selection conditions, defining CC containment, disjointness and intersection. In the case where there are no intersecting CCs and no disjunctions, we give an algorithm to complete \joinView\ that models the containment and disjointness of CCs as a \textit{Hasse diagram} \cite{williamson2002combinatorics} that it recurses on bottom-up to fill-in \joinView. Any leftover \joinView\ tuples without $B_1, \ldots, B_q$ values are randomly assigned a combination that cannot cause a new contribution towards the target count of any CC. However, if no such combinations are available, then the leftover \joinView\ tuples cannot be completed. We refer to these as \noJoin\ tuples.

\paratitle{Hybrid Approach (Section \ref{sec:hybrid_approach}):}
In the absence of intersecting CCs, the solution decomposes cleanly as seen above. This motivates the hybrid approach that combines ideas from both cases to achieve better runtime and accuracy when some CCs intersect. We start by labeling each pair of CCs as disjoint, contained or intersecting. 
For all CCs that do not intersect or contain any intersecting CCs, we use the approach from Section \ref{subsec:cc_relationships}, and for the rest, we use the ILP approach from Section \ref{subsec:cc_general}. Lastly, as seen above in the special case, we may end up with some \noJoin\ tuples.

\subsection{Overview of the Second Phase}
After filling-in the columns of \joinView\ that originate in $R_2$ ($B_1 \ldots, B_q$), we turn to reverse-engineering $R_1$ from \joinView. This phase uses \textit{conflict hypergraphs} \cite{ChuIP13} to represent possible DC violations.

\paratitle{Conflict Hypergraph (Section \ref{subsec:list_coloring}, red box in Figure \ref{fig:detailed}):}
We use the notion of conflict hypergraph for the tuples of $R_1$ based on the DCs. Given a DC, we construct an edge for all the sets of tuples that cannot get the same foreign key value due to that DC.


\begin{example}\label{ex:hypergraph_overview}
Consider the relation $R_1$ depicted in Figure \ref{fig:db_table} and the first DC in Figure \ref{fig:dcs}. Suppose the first two tuples are assigned the same $Area$ value in \joinView. Thus, the conflict hypergraph will have an edge containing the tuples with $p_{id} = 1$ and $p_{id} = 2$ since they are both owners and cannot be in the same household (the $h_{id}$ value). 
\reva{The conflict hypergraph of our running example is depicted in Figure \ref{fig:graph}.}
\end{example}

\paratitle{List Coloring (Section \ref{subsec:list_coloring}, orange box in Figure \ref{fig:detailed}):}
Proper coloring of the hypergraph ensures that there must be at least two vertices in each edge with distinct colors. Thus, modeling each $FK$ value as a color and each tuple as a vertex allows us to prove that a proper coloring results in an assignment of $FK$ values that satisfies the DCs. 
The values in \joinView\ filled-in by the previous phase induce a list of possible $FK$ values, and thus colors, for $R_1$ tuples. 
Finding a proper coloring such that each vertex assumes a color from its predefined list is called List Coloring \cite{achlioptas_molloy_1997} and is NP-hard. 
We thus propose a greedy coloring algorithm based on vertex degree. 

\paratitle{Algorithm for Satisfying the DCs (Section \ref{sec:dc_algo}):}  
The size of the conflict hypergraph can be very large 
and thus may cause a significant slowdown in practice. 
Therefore, we partition $R_1$ into smaller sets with the same $B_1 \ldots, B_q$ values and construct a conflict hypergraph for each set separately. 
For each non-\noJoin\ tuple, \joinView\ contains $B_1 \ldots, B_q$ values, 
so we can use our greedy coloring algorithm to find a coloring for them. 
We color \noJoin\ tuples at the end using all $FK$ values as candidates. 
This phase may result in the addition of extra tuples to $R_2$ (the second output in Figure \ref{fig:detailed}). 

%% file: cc_sol.tex
\section{First Phase: Solving CCs}\label{sec:cc}
In this section, we focus on the first phase. Given two relations $R_1(K_1,$ $A_1, \ldots, A_p, FK)$ and $R_2(K_2,B_1,\ldots, B_q)$, we wish to satisfy a set $S_{CC}$ of CCs over the join view $\joinView = R_1\bowtie_{FK=K_2}R_2$. 

\subsection{Solution as an ILP}
\label{subsec:cc_general}
We give a two-part solution in Algorithm \ref{algo:cc_intersection} where we: (1) model the CCs as a system of linear equations and solve it using an ILP solver, and (2) greedily fill-in $B_1, \ldots, B_q$ values for each tuple in \joinView. 
The first part (lines \ref{l:denote}--\ref{l:solve}) is inspired by \cite{Arasu2011}. Each variable represents the number of tuples with a specific combination of $A_1, \ldots, A_p, B_1, \ldots, B_q$ values in \joinView. 
Each CC is written as a sum of the variables whose associated tuples satisfy its selection condition. 
We now introduce the notion of {\em intervalization} \cite{Arasu2011}.

\paratitle{Intervalization:}
Creating a variable for every combination of values in the cross product of the full domains of all the $p+q$ (non-key) columns in \joinView\ would give a very large ILP. 
We augment the notion of intervalization \cite{Arasu2011} so that it will not only assist in reducing the number of variables based on the intervals of values in $S_{CC}$, but also use only the combinations of $A_1, \ldots, A_p$ values already in $R_1$.
We call this \textit{binning} the distinct $(A_1, \ldots, A_p)$ values in $R_1$. 

In the system of equations $Ax=b$, row $r_i$ (in $A$) corresponds to $CC_i$ and row $b_i$ (in $b$) stores $CC_i$'s target count. We create the vector $x$ of variables by putting bins with the same $B_1, \ldots, B_q$ values as contiguous elements (see Example~\ref{ex:ILP}). 
Since input CCs are linear, each element in $A$ is 
$0$ or $1$. 
The goal is to solve for an $x$ with non-negative integer entries (line \ref{l:solve}). Such a solution can be obtained if there exists a solution to \prob\ where $R_1\bowtie_{FK=K_2} R_2$ satisfies $S_{CC}$.
In the second part (lines \ref{l:iteratevals}--\ref{l:updatevals}), we fill-in the $B_1, \ldots, B_q$ values greedily. 
For each assignment $x_i = v_i$, we find at most $v_i$ tuples (with empty $B_1, \ldots, B_q$ cells) in \joinView\ that satisfy $x_i$'s selection condition on $R_1$, and fill-in their $B_1, \ldots, B_q$ values as encoded by $x_i$. 

\begin{example}
\label{ex:ILP}
    Reconsider relations $R_1$ and $R_2$ in Figure \ref{fig:db_table}, CCs in Figure \ref{fig:ccs} and \joinView\ 
    described in Example \ref{ex:join_view_missing}.
    Intervalization splits $Age$ into 
    $[0, 24]$ and $[25, 114]$ due to $CC_3$ (all other columns are categorical). 
    Even though $R_1$ contains multiple tuples for \reva{multi-lingual} homeowners with age greater than $24$, it suffices to look at those with $Age$ in $[0, 24]$ and $[25, 114]$. 
    Importantly, for the given instance, 
    we only need to keep track of the following tuple types: (1) $Age\in [25, 114]$, $Rel=$ Owner, \reva{$Multi$-$ling=0$}, (2) $Age\in [0, 24]$, $Rel=$ Spouse, \reva{$Multi$-$ling=0$}, (3) $Age\in [0, 24]$, $Rel=$ Child, \reva{$Multi$-$ling=1$}, and (4) $Age\in [25, 114]$, $Rel=$ Owner, \reva{$Multi$-$ling=1$}. Here, vector $x$ uses a copy of these four bins with $Area=$ Chicago in $x_1$ to $x_4$ and 
    $Area=$ NYC in $x_5$ to $x_8$. \revc{Without the idea of binning, we would need $16$ variables because $Area$ can take $2$ distinct values and $R_1$ contains $8$ unique tuples.} Finally, we iterate through each $CC_i\in S_{CC}$ and add rows $r_i$ and $b_i$ in $A$ and $b$, resp. \revc{For $CC_1$, $r_i=[1, 0, 0, 1, 0, 0, 0, 0]$ and $b_i=4$ 
    because only $x_1$ and $x_4$ match the selection conditions in $CC_1$; similarly for other CCs.} 
    Hence, $Ax=b$ 
    has a solution given by $x_1=2, x_2=1, x_3=2, x_4=2, x_5=1, x_6=0, x_7=0$ and $x_8=1$. 
    Finally, we iterate through $x_i$'s to find \joinView\ tuples which satisfy its selection condition and assign the matching $Area$ value that gives the view in Figure~\ref{fig:v1_complete}. E.g., we find two tuples in \joinView\ with $Age\in[25, 114]$, $Rel=$ Owner and \reva{$Multi$-$ling=0$} for $x_1$ and assign $Area=Chicago$. 
\end{example}

\begin{figure}[!htb]
	\begin{footnotesize}
	\centering
	\begin{tabular}{| c | c | c | c || c |}
		\hline $p_{id}$ & $Age$ & $Rel$ & \reva{$Multi$-$ling$} & $Area$ \\
		\hline $1$ & $75$ & $Owner$ & $0$ & $Chicago$\\
		\hline $2$ & $75$ & $Owner$ & $1$ & $Chicago$\\
		\hline $3$ & $25$ & $Owner$ & $0$ & $Chicago$\\
		\hline $4$ & $25$ & $Owner$ & $1$ & $Chicago$\\
		\hline $5$ & $24$ & $Spouse$ & $0$ & $Chicago$\\
		\hline $6$ & $10$ & $Child$ & $1$ & $Chicago$\\
		\hline $7$ & $10$ & $Child$ & $1$ & $Chicago$\\
		\hline $8$ & $30$ & $Owner$ & $0$ & $NYC$\\
		\hline $9$ & $30$ & $Owner$ & $1$ & $NYC$\\
		\hline
	\end{tabular}
	\caption{Join view $\joinView = (R_1\bowtie_{FK=K_2}R_2)$ of $R_1$ and $R_2$ from Figure \ref{fig:db_table} with filled-in $Area$ values}\label{fig:v1_complete}
	\end{footnotesize}
\end{figure}



\begin{algorithm}[ht]
    \caption{Complete \joinView\ - Intersecting CCs}\label{algo:cc_intersection}
    
    \SetKwInOut{Input}{Input}\SetKwInOut{Output}{Output}
    \LinesNumbered
    \Input{Relations $R_1(K_1, A_1, \ldots, A_p)$ and $R_2(K_2, B_1, \ldots, B_q)$, 
    $S_{CC}$ - set of linear CCs 
    with target counts}
    \Output{\joinView\ - $B_1, \ldots, B_q$ values filled-in 
    }
    \BlankLine
    
    \tcc{model CCs as integer program and solve}
    View $\joinView(K_1, A_1, \ldots, A_p, B_1, \ldots, B_q) \gets$ copy of $R_1$ with empty $B_1, \ldots, B_q$ columns\;
    $n_{\joinView} \gets$ \revc{number of bins in which distinct tuples of $R_1$ are grouped using binning}\;\label{l:denote}
    $\forall i\in [q],~ n_i \gets$ number of distinct $B_i$ values in $R_2$\;
    $n \gets n_{\joinView}\times\prod\limits_{i=1}^q n_i$\;
    \revc{/*$A$ will be a $(n_{\joinView}+|S_{CC}|)\times n$ matrix for CCs*/}\label{l:what_is_A}\\
    $b \gets$ empty $(n_{\joinView}+|S_{CC}|)\times 1$ vector for target counts\;
    $x \gets n\times 1$ vector for non-negative integer variables\;\label{l:elements_in_x}
    \For{each tuple type $t_i$ accounted for by $n_{\joinView}$}
    {   \label{l:augment_S_CC}
        Add row in $A$ \revc{/* $0$'s except $1$ for relevant variables in $x$*/}\;\label{l:createrow_augment_S_CC}
        $b[i] \gets $ number of copies of $t_i$ in $R_1$\;
    }
    \For{each $CC_i\in S_{CC}$}
    {
        Add row in $A$ \revc{/* $0$'s except $1$ for relevant variables in $x$*/}\;\label{l:createrow}
        $b[i] \gets CC_i.target$\;
    }
    Compute $x$ by solving $Ax=b$\;\label{l:solve}
    \tcc{fill values in $B_1, \ldots, B_q$ greedily}
    \For{each $x_i\in x$ with value $c_i$}
    {\label{l:iteratevals}
        Find (at most) $c_i$ tuples satisfying $x_i$'s condition in \joinView\;
        Update $B_1, \ldots, B_q$ values encoded by $x_i$\;\label{l:updatevals}
    }
    \Return $\joinView$
\end{algorithm}
\paratitle{Augmenting with All-Way Marginals:}
When $A$ is sparse, some $x_i$ values in the solution may not match the true counts. 
Despite such discrepancies, we can complete several tuples in $\joinView$ because we update at most as many tuples as the value of $x_i$ in the solution. 
The order of updates may also impact which subset of $\joinView$ tuples gets specific $B_1, \ldots, B_q$ values. 
For example, another solution to the ILP in Example~\ref{ex:ILP} is given by $x_1=0, x_2=3, x_3=0, x_4=4, x_5=x_6=x_7=0$, $x_8=2$. This assigns $Area=$ Chicago to tuples with $p_{id}=2, 4, 5, 9$ in \joinView. However, the remaining tuples do not get any $Area$ value and no CC in $S_{CC}$ gets satisfied in \joinView. 
We overcome this issue by using both $S_{CC}$ and all all-way marginals over $A_1, \ldots, A_p$ from $R_1$ when solving the ILP (see the discussion about the baseline's CC accuracy in Section~\ref{sec:experiments}). The solution reported in Example~\ref{ex:ILP} was computed with all all-way marginals.

\paratitle{Complexity:} The complexity of Algorithm~\ref{algo:cc_intersection} is $O(|S'_{CC}|\cdot m + S)$, where $S'_{CC}$ contains CCs from $S_{CC}$ along with the marginals, and $m$ is the number of variables that is upper-bounded by the number of tuples in $R_1$ times the product of the sizes of the active domains of $B_1, \ldots, B_q$ in $R_2$. Lastly, $S$ is the time complexity of the ILP solver. 

\subsection{Efficient Algorithm for Special CC Types}\label{subsec:cc_relationships}
In practice, Algorithm \ref{algo:cc_intersection} may incur slow runtimes as generating and solving the system of equations is time consuming, even with state-of-the-art ILP solvers (as shown in Section \ref{sec:experiments}). Thus, we describe a model for 
relationships between the CCs in $S_{CC}$ and devise an algorithm to better tackle \joinView\ completion in specific scenarios. 


\begin{figure}[ht]
    \begin{small}
    \centering
        \centering
        \begin{enumerate}
        \itemsep0em
        \item[] $CC_1: |\sigma_{Age\in[10, 14], Area=Chicago} (R_1\bowtie R_2)| = 20$
        \item[] $CC_2: |\sigma_{Age\in[50, 60], \reva{Multi\text{-}ling=0}, Area=NYC} (R_1\bowtie R_2)| = 25$
        \item[] $CC_3: |\sigma_{Age\in[13, 64], Area=Chicago} R_1\bowtie R_2| = 100$
        \item[] $CC_4: |\sigma_{Age\in[18, 24], \reva{Multi\text{-}ling=0}, Area=Chicago} R_1\bowtie R_2| = 16$
        \end{enumerate}
    \end{small}
    \caption{CC relationships. $CC_1 \cap CC_2 = \emptyset$, and $CC_4 \subseteq CC_3$}\label{fig:cc_relationships}
    \vspace{-4mm}
\end{figure}



\begin{definition}
\label{def:disjoint}
$CC_i, CC_j\in S_{CC}$ are disjoint either if their selection conditions on the $R_1$ attributes are disjoint, or if their selection conditions on $R_1$ are identical and the conditions on $R_2$ are disjoint. We denote this by $CC_i\cap CC_j = \emptyset$.
\end{definition}


Note that we also consider pairs of CCs with the same $R_1$, but disjoint $R_2$ selection conditions as disjoint. 
For a pair $(CC_i, CC_j)$ of such CCs, assigning $B_1, \ldots, B_q$ values in tuples that contribute to the count of $CC_i$ should not limit the set of tuples available for $CC_j$, if a solution exists. We label such pairs similarly to a pair of disjoint CCs. 
Next, we define the notion of CC containment.

\begin{definition}
\label{def:contained}
Let $CC_i, CC_j \in S_{CC}$ such that $CC_i: |\sigma_{\varphi_i}(R)| = k_i$ and $CC_j: |\sigma_{\varphi_j}(R)| = k_j$. $CC_i$ is contained in $CC_j$, denoted $CC_i \subseteq CC_j$, if $\varphi_i$ uses a (non-strict) superset of attributes in $\varphi_j$ and for each common attribute, the values in $CC_i$ are a subset of the corresponding values in $CC_j$.
\end{definition}


Intuitively, if $CC_i$ is contained in $CC_j$, then $CC_i$ is more restrictive than $CC_j$, and assigning a tuple $t\in R_1$ values in $B_1, \ldots, B_q$ that satisfy the selection condition in $CC_i$ will also satisfy the selection condition in $CC_j$. This observation defines a partial order on $S_{CC}$ which we utilize later to find a solution for CCs.

\begin{definition}
\label{def:intersect}
    $CC_i, CC_j\in S_{CC}$ are said to be intersecting if they are neither disjoint nor does one contain the other. 
    We denote this by $CC_i\cap CC_j \neq \emptyset$.
\end{definition}


\begin{example}
    \label{ex:intersect}
    Assume $R_1$ (or $\joinView$) contains $10$ tuples with $Age\in[10, 30)$, $20$ with $Age\in[30, 50)$ and $50$ with $Age\in[50, 70]$. Let: 
    \begin{enumerate}
        \itemsep0em
        \item[] $CC_1: |\sigma_{Age\in[10,50), Area=Chicago} R_1\bowtie R_2| = 30$
        \item[] $CC_2: |\sigma_{Age\in[30, 70], Area=NYC} R_1\bowtie R_2| = 30$
        \end{enumerate}
    If all tuples with $Age\in[30, 50)$ get assigned $Area=NYC$, $CC_1$ cannot be satisfied.
    Even when $Area=Chicago$ in $CC_2$, it is unclear how many tuples with age in $[30, 50)$ can be assigned $Area=Chicago$.
\end{example}


\paratitle{Solution Without Intersecting CCs:}
Now, we focus on the setting where there are no intersecting CCs present and describe Algorithm~\ref{algo:cc_no_intersection} that outputs an exact solution. 

We use the notion of a Hasse diagram \cite{williamson2002combinatorics}, denoted by $\mathcal{H}=(V, E)$, to encode the containment relationships between the CCs in $S_{CC}$. We refer to each connected component in the undirected version of $\mathcal{H}$ as a diagram. Within each diagram, the CC that is not contained in any other CC is referred to as the \textit{maximal element}. 


Algorithm \ref{algo:cc_no_intersection} is given the join view \joinView\ with missing $B_1,\ldots, B_q$ columns, $S_{CC}$ and the Hasse diagram $\mathcal{H}$ describing the containment relations in $S_{CC}$. We denote by $\mathcal{V}(\mathcal{H})$ and $\mathcal{E}(\mathcal{H})$ the collective set of all nodes and edges of the diagrams in $\mathcal{H}$. 
The algorithm operates recursively with a single base case -- if all the CCs in $S_{CC}$ are disjoint, i.e., $\mathcal{E}(\mathcal{H})$ is empty (line \ref{l:all_disjoint}), then it simply chooses $k_i$ tuples that can contribute to each $CC_i \in S_{CC}$ and completes their $B_1,\ldots, B_q$ values given by $CC_i$. 
When the base case is not met, for each $H\in\mathcal{H}$, the algorithm makes a recursive call on each child of the maximal element $m$ in $H$ (lines \ref{l:iter_child}--\ref{l:call_recurse}) to get the resulting view of the sub-diagram and then finds the remaining number of tuples that will get $CC_m$ to its target count (lines \ref{l:combine_with_root}--\ref{l:combine}). 
Finally, in the loop in line \ref{l:remaining}, the algorithm completes any missing values in the tuples while ensuring that these values do not add to the count of any $CC\in S_{CC}$ by finding combinations that are not specified in $S_{CC}$.

\begin{algorithm}[ht]
    \caption{Complete \joinView\ - Non-intersecting CCs}
    \label{algo:cc_no_intersection}
    
    \SetKwInOut{Input}{Input}\SetKwInOut{Output}{Output}
    \LinesNumbered
    \Input{\joinView\ - View to complete, 
        $S_{CC}$ - Set of CCs, 
        $\mathcal{H}$ - Set of diagrams encoding CC containment}
    \Output{\joinView\ - $B_1, \ldots, B_q$ values filled-in}
    
    \BlankLine
    
    $\forall i\in \mathcal{V}(\mathcal{H}).~\sigma_i, k_i\leftarrow$ selection condition on $R_1$, count\;\label{l:first_non_inter}
    \If{$\mathcal{E}(\mathcal{H})=\emptyset$}
    {\label{l:all_disjoint}
        \ForEach{$i\in \mathcal{V}(\mathcal{H})$}
        {\label{l:assign_vals_loop}
            Find $k_i$ tuples in \joinView\ (without $B_1,\ldots,B_q$ values) that satisfy $\sigma_i$\;\label{l:findTo_assign_vals_end}
            Assign $B_1, \ldots, B_q$ values\;\label{l:assign_vals_end}
        }
        \Return \joinView \;\label{l:return_all_disjoint}
    }
    \ForEach{$H\in \mathcal{H}$}
    {\label{l:loop_diagrams}
        $m\leftarrow$ maximal elem. in $H$\;\label{l:maximal}
        \ForEach{$c \in children(m)$}
        {\label{l:iter_child}
            $H_{c}\leftarrow$ sub-diagram with maximal elem. $c$\;
            $\joinView =$ Algorithm \ref{algo:cc_no_intersection}($\joinView, S_{CC}, \{H_{c}\}$)\;\label{l:call_recurse}
        }
        Find $k_{m}-\sum_{c \in children(m)}k_c$ tuples in \joinView\ that satisfy $\sigma_m\bigwedge\limits_{c \in children(m)}\neg\sigma_c$\;\label{l:combine_with_root}
        Assign $B_1, \ldots, B_q$ values from $CC_m$\;\label{l:combine}
    }
    \revc{$combo_{unused} \gets$ list of combinations in $R_2$ columns that are not relevant to $S_{CC}$}\;\label{l:combo_unused}
    \ForEach{$t\in \joinView$}
    {\label{l:remaining}
        \If{$t.B_{i}, \ldots, t.B_{q}$ values are missing}
        {\label{l:add_combos}
            Assign a combination of values \revc{from $combo_{unused}$}\;\label{l:complete_partial_assignment}
        }
    }
    \Return \joinView\;
\end{algorithm}

\begin{example}\label{ex:non_intersection_algo}
Reconsider CCs 1--4 in Figure \ref{fig:cc_relationships}. The set $\mathcal{H}$ is $\{H_1, H_2,$ $H_3\}$, where $H_1$ and $H_2$ contain only $CC_1$ and $CC_2$, respectively, and $H_3$ is a diagram composed of one edge from $CC_3$ to $CC_4$. 
Algorithm \ref{algo:cc_no_intersection} gets $\mathcal{H}$ along with \joinView\ and $S_{CC}$ as input. It 
assigns $CC_i$'s selection condition on $R_1$ and target count to $\sigma_i$ and $k_i$, for all $i$. Then, it checks the condition in line \ref{l:all_disjoint}, which does not hold as we have the edge $(CC_3, CC_4)$. 
Thus, it goes to the loop in line \ref{l:loop_diagrams} to iterate over the three diagrams. 
For $H_3$, the maximal element is $CC_3$, so Algorithm \ref{algo:cc_no_intersection} recursively calls itself for the sub-diagram containing only $CC_4$ (line \ref{l:call_recurse}) and finds $16$ tuples such that $Age\in[18,24]$ and \reva{$Multi$-$ling=0$}, and assigns to them $Area=Chicago$ (lines \ref{l:assign_vals_loop}--\ref{l:assign_vals_end}). It then returns from the recursive call to find $100-16=84$ tuples with 
$Age\in[13,64]\setminus[18,24]$ and \reva{$Multi$-$ling\neq 0$}
, and assigns to them $Area = Chicago$ (lines \ref{l:combine_with_root}--\ref{l:combine}). 
For $H_1$ ($H_2$) the maximal element is $CC_1$ ($CC_2$), the algorithm then performs a recursive call to itself with an empty diagram, and returns from the call to select $20$ ($25$) tuples that have $Age\in [10,14]$ ($Age\in [50,60]$ and \reva{$Multi$-$ling=0$}) and assign to them $Area=Chicago$ ($Area=NYC$). \revc{
Here, $combo_{unused}$ contains values from $Area$'s domain except $Chicago$ and $NYC$ which get used in $S_{CC}$. If there are any tuples in \joinView\ without an assignment (see loop on line~\ref{l:remaining}), we assign to each a value chosen from $combo_{unused}$.} 
\end{example}



At the end of the algorithm, any tuple in \joinView\ without $B_1, \ldots, B_q$ values is randomly assigned a combination of values that is not used in $S_{CC}$ (line \ref{l:complete_partial_assignment}). We refer to these tuples as {\em \noJoin\ tuples} if no such combination is available. Observe that the matching tuples in $R_1$ do not have an $FK$ to $K_2$ mapping
\revc{, i.e., if there is a tuple in \joinView\ that is missing an assignment, also called an \noJoin\ tuple, then \joinView\ does not give a set of candidate $K_2$ values that could be assigned in its $FK$ cell.
}
We will handle such tuples in Section \ref{sec:dc_algo}. 

\begin{proposition}\label{prop:hybrid}
    If $S_{CC}$ does not contain intersecting CCs and there exists a join view \joinView\ that satisfies all CCs in $S_{CC}$, then Algorithm \ref{algo:cc_no_intersection} finds such a view.
\end{proposition}

\common{
\paratitle{Complexity:}
The complexity of Algorithm \ref{algo:cc_no_intersection} is $O(|S_{CC}|^2\cdot d_1 + |S_{CC}| \cdot (\max_{i} |dom_a(B_i)|)^{d_2} + |S_{CC}|\cdot |V_{Join}|)$
, where $d_1, d_2$ are the number of columns in \joinView\ and $R_2$, $dom_a(B_i)$ is the active domain of $R_2.B_i$. The first term is for computing the relationships between CCs and recursing on the Hasse diagrams
(lines \ref{l:loop_diagrams}--\ref{l:call_recurse}), second term is for constructing $combo_{unused}$ (line \ref{l:combo_unused}) and third term is for lines \ref{l:first_non_inter}--\ref{l:return_all_disjoint}, \ref{l:combine_with_root}--\ref{l:combine} and choosing a random value per tuple in lines \ref{l:remaining}--\ref{l:complete_partial_assignment}. In practice, we only consider columns used in $S_{CC}$ instead of $d_2$.
}




\subsection{Hybrid Approach}
\label{sec:hybrid_approach}
In many cases, $S_{CC}$ contains a combination of disjoint, contained, and intersecting CCs, so we combine Algorithms \ref{algo:cc_intersection} and \ref{algo:cc_no_intersection}.



We start by constructing Hasse diagram based on containment relationship between pairs of CCs in $S_{CC}$. Next, we iterate through each diagram $H\in\mathcal{H}$, and discard $H$ if it contains \revc{intersecting CCs}. 
\revc{Note that the absence of an edge in the Hasse diagram does not guarantee the lack of intersection at the beginning of phase I 
(demonstrated by Example \ref{ex:intersect} where the Hasse diagram starts out as two nodes without an edge, but the CCs represented by these nodes do intersect). 
Therefore, we keep track of which CCs intersect to then discard the affected diagrams (set $S_2$) and run Algorithm \ref{algo:cc_no_intersection} on the remaining diagrams (set $S_1$).} 
In particular, $\forall CC_i\in S_1, CC_j\in S_2.~CC_i\cap CC_j = \emptyset$, $CC_i\not\subseteq CC_j$ and $CC_j\not\subseteq CC_i$. We then run Algorithm \ref{algo:cc_no_intersection} for CCs in $S_1$, and Algorithm \ref{algo:cc_intersection} for those in $S_2$. 

As seen above, it is possible that some tuples may not have a $B_1, \ldots, B_q$ assignment in \joinView. Let $S_3$ be the set of these tuples that are dealt with using $combo_{unused}$ as described in Example~\ref{ex:non_intersection_algo}. 
If $|combo_{unused}|=\emptyset$, then all tuples in $S_3$ are \noJoin\ tuples. 

\paratitle{Augmenting with Modified Marginals}
Our approach guarantees that the partial solution returned by Algorithm~\ref{algo:cc_no_intersection} satisfies $S_1$ exactly. 
In comparison to how we augment $S_{CC}$ with marginals in Section~\ref{subsec:cc_general} before solving the ILP, we now want the scope of the marginals being added to be limited to the tuples that are relevant for the CCs in $S_2$. For example, let $S_{CC}=\{CC_1, CC_3\}$ from Figure~\ref{fig:ccs}. We add CCs with the following selection predicates: (1) $Age<=24, Rel=Owner,$ \reva{$Multi$-$ling=0$}, and (2) $Age<=24, Rel=Owner,$ \reva{$Multi$-$ling=1$}. 
It may still happen that the matrix $A$ is sparse and some $x_i$'s do not match the true counts causing some CC errors.

%% file: dc_sol.tex
\section{Second Phase: Adding DCs}
\label{sec:dc}
We start by presenting our model for conflict hypergraph for FK DCs and then use it to describe the solution for DCs. 
In short, our approach is to reverse-engineer $R_1$ from \joinView\ so that joining it with $R_2$ recovers \joinView, and $R_1$ satisfies all DCs in $S_{DC}$. 
\begin{figure}[t]
\centering
\begin{tikzpicture}[shorten >=1pt,-,scale=0.65]
  \tikzstyle{vertex}=[circle,fill=blue!25,minimum size=12pt,inner sep=0pt]
  \foreach \name/\angle/\text in {P-4/0/4, P-6/40/6, P-2/80/2, P-7/120/7, P-1/160/1, P-3/200/3, P-8/240/8, P-9/280/9, P-5/320/5}
    \node[vertex,xshift=8cm,yshift=8cm] (\name) at (\angle:1.5cm) {$\text$};
  \draw (P-1) -- (P-2);
  \draw (P-1) -- (P-3);
  \draw (P-1) -- (P-4);
  \draw (P-2) -- (P-3);
  \draw (P-2) -- (P-4);
  \draw (P-3) -- (P-4);
  \draw (P-1) -- (P-5);
  \draw (P-2) -- (P-5);
  \draw (P-2) -- (P-6);
  \draw (P-2) -- (P-7);
  \draw (P-9) -- (P-8);
  \draw[dashed] (P-9) -- (P-1);
  \draw[dashed] (P-9) -- (P-2);
  \draw[dashed] (P-9) -- (P-3);
  \draw[dashed] (P-9) -- (P-4);
  \draw[dashed] (P-8) -- (P-1);
  \draw[dashed] (P-8) -- (P-2);
  \draw[dashed] (P-8) -- (P-3);
  \draw[dashed] (P-8) -- (P-4);
\end{tikzpicture}
\caption{Conflict graph for the tuples in $R_1$ from Figure~\ref{fig:db_table} based on $\joinView$ from Figure~\ref{fig:v1_complete}. Two tuples are connected by a solid edge if their $Area$ values match but they cannot be assigned the same $h_{id}$ value. Dashed edges show DC violations between tuples with different $Area$ values. Partitioning \joinView\ by $Area$ values addresses such violations because the set of household ids is disjoint for different $Area$ values}
\label{fig:graph}
\vspace{-3mm}
\end{figure}
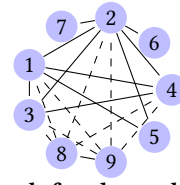

\subsection{Conflict Hypergraphs and List Coloring}\label{subsec:list_coloring}
We slightly augment the notion of conflict hypergraphs to illustrate possible Foreign Key DC violations caused by subsets of $R_1$ tuples. 

\begin{definition}[Conflict Hypergraph for Foreign Key DCs]
A conflict hypergraph for $R_1$ and $S_{DC}$ is defined as $G = (V,E)$ where $V$ is the set of tuples in $R_1$ and $e = \{t_1,\ldots, t_k\} \in E$ if there is a Foreign Key DC of the form $\neg (\varphi(t_1,\ldots, t_k) \land t_1.FK = \ldots = t_k.FK)$ such that $\varphi(t_1,\ldots, t_k)$ evaluates to $True$. 
\end{definition}
It suffices to consider only $\varphi(t_1,\ldots, t_k)$ in the DCs when adding edges because $FK$ is initially missing. 
Abusing notation, we denote a set of tuples $\mathcal{T}$ violating $\varphi_i(t_1,\ldots, t_k)$ in a given DC $\sigma_i$ by $\mathcal{T} \not\vDash_{\varphi_i} \sigma_i$ (we will use this notation in Algorithm \ref{algo:two_partition}). 


Next, we give the connection between conflict hypergraph coloring and $FK$ assignment in $R_1$, so a proper coloring satisfies DCs.

\begin{proposition}\label{prop:color_satisfies}
Given an instance of \prob, a coloring of the conflict hypergraph gives an assignment to all cells of the missing FK column in $R_1$ such that all DCs are satisfied.
\end{proposition}

We now turn to the problem of inferring $FK$ values in $R_1$ from the completed \joinView. 
Each tuple in $R_1$ can have multiple options for foreign key values that lead to \joinView\ obtained in phase \rom{1}. This establishes a list of possible colors, also referred to as \textit{candidate colors}, for each vertex in the conflict graph.
This problem is called List Coloring \cite{jensen2011graph,achlioptas_molloy_1997}. It is a generalization of $k$-coloring, and is 
thus NP-hard. Hence, we use a heuristic approach, described by Algorithm~\ref{algo:lf_coloring}, to color the vertices in a non-increasing order by degree, coloring as many vertices as possible. 
In Section \ref{sec:dc_algo}, we describe how to color the vertices that remain uncolored by Algorithm~\ref{algo:lf_coloring}.

Algorithm~\ref{algo:lf_coloring} takes as input the conflict hypergraph $G_c$\revc{, a mapping $c$ from vertices to colors (initially empty) and a list of candidate colors $L$}. It can be called on a graph with a partial color assignment (used in Algorithm~\ref{algo:two_partition} in Section \ref{sec:dc_algo}). 
Initially, $s$ is an empty list that is used to store skipped vertices, and $l$ is the list of uncolored vertices sorted in non-increasing order by degree in $G_c$ (lines \ref{l:s}--\ref{l:sort}). 
\revc{In lines \ref{l:iteratel}--\ref{l:addtos}, we find a list of permissible colors for each $v\in l$, i.e., those vertices in $G_c$ that have not been given a color in the input color map $c$. If vertex $v$ belongs to an edge $e$ where all vertices other than $v$ in $e$ have the same color $c$, then $c$ is a forbidden color for $v$.}
Next, 
the algorithm assigns the ``smallest'' available color to $v$ in line~\ref{l:assign}. Otherwise, $v$ gets added to $s$ and remains uncolored (line \ref{l:addtos}). Finally, 
color map $c$ and 
list of skipped vertices $s$ are returned.

\IncMargin{1em}
\begin{algorithm}[!ht]
    \caption{Largest-first list coloring}
    \label{algo:lf_coloring}
    
    \SetKwInOut{Input}{Input}\SetKwInOut{Output}{Output}
    \LinesNumbered
    \Input{$G_c$ - conflict hypergraph with color choices per vertex, $c$ - a map from vertices to colors so far\revc{, $L$ - list of candidate colors}}
    \Output{$c$ - updated coloring that builds on the input coloring, $s$ - list of skipped vertices}
    
    \BlankLine
    
    \SetKwFunction{FMain}{ColoringLF}
    \SetKwProg{Fn}{Function}{:}{}
    \Fn{\FMain{$G_c, c, \revc{L}$}}{
        $s\gets \emptyset$\;\label{l:s}
        $l \gets sortDescendingDeg(\{v\in V[G_c]~|~ v\not\in c\})$\;\label{l:sort}
        \For{$v \in l$}
        {\label{l:iteratel}
            $forbidden\gets \emptyset$\;
            \For{$e$ s.t. $v \in e$}
            {
                \revc{
                \If{$\exists c~\forall u \neq v \in e.~ c[u] = c$} 
                {\label{l:all_colored_the_same}
                    $forbidden.add(c)$\;\label{l:add_forbidden}
                }
                }
            }
            
            \If{$\revc{L \setminus forbidden} \neq \emptyset$}
            {\label{l:if_color_avail}
                $c[v] \gets \min(\revc{L \setminus forbidden})$\;\label{l:assign}
            }
            \Else{
            $s \gets s \cup \{v\}$\;\label{l:addtos}
            }
        }
        \Return $c$, $s$
    }
\end{algorithm}
\DecMargin{1em}

\begin{example}
Reconsider the view \joinView\ and DCs shown in Figure~\ref{fig:v1_complete} and \ref{fig:dcs}, respectively. Figure~\ref{fig:graph} (including the dashed edges) gives the resulting conflict graph $G_c$. For example, there is an edge between vertices $1$ and $2$ because $t_1.Rel=t2.Rel=$ Owner, so assigning them the same $FK$ value would violate $DC_{O, O}$. Here, $l=[2, 1, 3, 4, 8, 9, 5, 6, 7]$. Thus, Algorithm~\ref{algo:lf_coloring} returns: $c[1]=2, c[2]=1, c[3]=3, c[4]=4, c[5]=3, c[6]=2, c[7]=2, c[8]=5$ and $c[9]=6$.
\label{ex:colorvalidparts}
\end{example}

\paratitle{Complexity:}
The complexity of Algorithm \ref{algo:lf_coloring} is $O(|V|\cdot\log{|V|} + |V|\cdot |E|)$ since, the algorithm sorts all vertices by degree (line \ref{l:sort}) and then traverses all edges adjacent to each vertex.

\subsection{Algorithm for DCs}\label{sec:dc_algo}
We describe Algorithm \ref{algo:two_partition} as the last step in solving \prob\ by completing $R_1.FK$. 
In Section \ref{sec:cc}, we showed how to complete \joinView\ by assigning values in the columns that came from $R_2$. 
For a tuple $t\in$\joinView\ with values $t.B_i=b_i, 1\leq i\leq q$, the candidate $FK$ values for the corresponding tuple in $R_1$ are given by $\pi_{K_2}\sigma_{B_1=b_1, \ldots, B_q=b_q} R_2$. We begin with an optimization that we employ in the algorithm. 

\paratitle{Optimization:}
Working with a single conflict hypergraph when $R_1$ contains a large number of tuples would not scale since the hypergraph can form one clique in the worst-case. However, observe that we can partition the filled-in \joinView\ and $R_2$ by $B_1, \ldots, B_q$ values into sets, and only consider conflict hypergraphs within each set because the candidate $FK$ values are disjoint across sets. 



\begin{example}
Reconsider relations $R_1$ and $R_2$, and view $\joinView$ shown in Figures~\ref{fig:db_table} and \ref{fig:v1_complete}, respectively. In $\joinView$, the tuples have $Area=Chicago$ or $Area=NYC$. Note that the set of candidate keys for tuples with $Area=Chicago$ comprises of values in $\pi_{h_{id}}\sigma_{Area=Chicago} R_2$ that is disjoint from those in $\pi_{h_{id}}\sigma_{Area=NYC} R_2$. This eliminates edges that would have been added to the conflict graph if we were to consider all vertices at once (shown as dashed edges in Figure~\ref{fig:graph}).
After partitioning \joinView\ 
by $B_1, \ldots, B_q$ values and using the DCs in Figure~\ref{fig:dcs}, we get two conflict graphs: (1) with vertices for tuples $t_1, \ldots, t_7$, and (2) with vertices for tuples $t_8$ and $t_9$. There is an edge between a {\bf pair} of vertices when the corresponding tuples {\bf would} violate a DC if assigned the same $h_{id}$ value, and these are shown as solid edges in Figure~\ref{fig:graph}.
\label{ex:conflictG}
\end{example}

\IncMargin{1em}
\begin{algorithm}[ht]
    \SetKwInOut{Input}{Input}\SetKwInOut{Output}{Output}
    \LinesNumbered
    \Input{View $\joinView(K_1, A_1, \ldots, A_p, B_1, \ldots, B_q)$, 
    Relations $R_1(K_1, A_1, \ldots, A_p, FK)$ and $R_2(K_2, B_1, \ldots, B_q)$, 
    $S_{DC}$ - set of DCs on $R_1$}
    \Output{$\hat{R_1}$ - copy of $R_1$ with $FK$ column filled-in, $\hat{R_2}$ - updated copy of $R_2(K_2, B_1, \ldots, B_q)$}
    
    \BlankLine
    $c_{all} \gets \emptyset$, $\hat{R_1}\gets$ copy of $R_1$, $\hat{R_2}\gets$ copy of $R_2$\;
    \For{$v = (b_1, \ldots, b_q) \in \pi_{B_1, \ldots, B_q}\joinView$}
    {\label{l:start1}
        $P_v = \{t\in \joinView~|~ \forall 1\leq i\leq q.~ t.B_i = b_i\}$\;
        $V \gets \emptyset$, $E \gets \emptyset$, $c\leftarrow \emptyset$\;
        $L=\pi_{K_2}\sigma_{B_1=b_1, \ldots, B_q=b_q} \hat{R_2}$\; 
        \For{$t_j \in P_v$}
        {
            $V \gets V \cup \{v_j\}$\; 
        }
        \For{$\mathcal{T} \subseteq P_v$ s.t. $\exists \sigma \in S_{DC}.~ \mathcal{T} \not\vDash_{\varphi} \sigma$}
        {\label{l:tuple_subsets}
            $E \gets E \cup \{\mathcal{T}\}$\;\label{l:add_edge}
        }
        $c, s \gets ColoringLF(G_c=(V, E), c, \revc{L})$\;\label{l:color}
        $L_{new} \gets |s|$ number of new colors\;\label{l:new_colors_for_s}
        $c, s \gets ColoringLF(G_c=(V, E), c, L_{new})$\;\label{l:color_s}
        \For{color $c_{new}$ in $L_{new}$ that gets used}
        {\label{l:add_row_R2}
            Add tuple $t_{new}$ in $\hat{R_2}$ with $t_{new}.K_2=c_{new}$ and $t_{new}.B_i=b_i$ ($1\leq i\leq q$)\;\label{l:create_row_R2}
        }
        $c_{all} \gets c_{all} \cup c$\;\label{l:end1}
    }
    $c_{all} \gets solveInvalidTuples(\joinView, S_{DC}, S_{CC}, \hat{R_2})$ 
    \;\label{l:handle_invalid}
    $\forall t_j\in \joinView, t'\in \hat{R_1}.~t_j.K_1=t'.K_1$ set $t'.FK=c_{all}[v_j]$\;\label{l:build_r1}
    \Return $\hat{R_1}$, $\hat{R_2}$\;\label{l:output_r1}
\caption{Complete $R_1.FK$ column using \joinView}
\label{algo:two_partition}
\end{algorithm} \DecMargin{1em}

Algorithm \ref{algo:two_partition} gets as input the view $\joinView$ outputted by the algorithm described in Section \ref{sec:hybrid_approach}, 
relations $R_1$ (with missing $FK$ values) and $R_2$, and set $S_{DC}$. 
It outputs $\hat{R_1}$, i.e., $R_1$ with values in the $FK$ column and $\hat{R_2}$, i.e., $R_2$ with possible additional tuples (as described next). 
The algorithm can be divided into three parts: (1) coloring the tuples that were assigned $B_1, \ldots, B_q$ values in $\joinView$, (2) coloring the \noJoin\ tuples, i.e., tuples that were {\em not} assigned $B_1, \ldots, B_q$ values in $\joinView$, and (3) coloring any skipped tuples (defined in Section \ref{subsec:list_coloring}). Algorithm \ref{algo:two_partition} maintains a map from tuples to their 
list of colors in $c$ and tracks the overall coloring in $c_{all}$. Eventually, $c_{all}$ has a color for every vertex that is used to complete $FK$ in $\hat{R_1}$ (lines \ref{l:build_r1}, \ref{l:output_r1}). 

In lines \ref{l:start1}--\ref{l:end1}, the algorithm iterates over each set of tuples with the same $B_1, \ldots, B_q$ value. Given a vector $v = (b_1, \ldots, b_q)$ of $q$ constants, it iterates over tuples in sets given by $P_v = \{t\in \joinView~|~ \forall 1\leq i\leq q.~ t.B_i = b_i\}$. For each $P_v$, it generates the conflict hypergraph $G_c$ as follows: a node $v_j$ per tuple $t_j$, 
a list $L$ of candidate colors given by the 
keys from tuples in $\hat{R_2}$  
with values $v_j.B_1, \ldots, v_j.B_q$, 
and an edge per set of tuples in $P_v$ that violates $\varphi$ in some DC. Next, \revc{$G_c$, $c$ and $L$} are inputted to Algorithm~\ref{algo:lf_coloring}, which outputs a partial coloring $c$ and a list of skipped vertices $s$. Vertices in $s$ are colored using at most $|s|$ new colors (lines \ref{l:new_colors_for_s}-\ref{l:create_row_R2}), resulting in insertion of tuples in $\hat{R_2}$ because colors correspond to primary keys in $\hat{R_2}$. 

Procedure $solveInvalidTuples$ (line \ref{l:handle_invalid}) handles the \noJoin\ tuples (defined in Section~\ref{subsec:cc_relationships}). 
Since these do not have $B_1, \ldots, B_q$ values in \joinView, the corresponding $\hat{R_1}$ tuples are missing $FK$ values because we have not yet considered them in any conflict hypergraphs. 
We construct a hypergraph for tuples in \joinView\ with edges incident to only the vertices for \noJoin\ tuples. 
However, the set $s$ outputted by Algorithm \ref{algo:lf_coloring} may contain \noJoin\ tuples that had to be skipped for a lack of available colors. 
Our strategy for coloring these 
is to assign to each a combination of $B_1, \ldots, B_q$ values that minimizes the error stemming from the CCs (defined in Section~\ref{sec:experiments}), 
and generate a tuple in $\hat{R_2}$ with a fresh key and the chosen $B_1, \ldots, B_q$ values. 
Finally, $\hat{R_1}.FK$ values are assigned based on the coloring $c_{all}$ (line \ref{l:build_r1}).

\begin{proposition}\label{prop:solution}
Given \joinView, $R_1$, $R_2$, $S_{DC}$, Algorithm~\ref{algo:two_partition} outputs relations $\hat{R_1}, \hat{R_2}$ such that $\hat{R_2}$ is a copy of $R_2$, possibly with more tuples, and $\hat{R_1}$ is a copy of $R_1$ with all the values in the $FK$ column completed such that $\forall \sigma\in S_{DC},~ \hat{R_1} \vDash \sigma$, and $\hat{R_1}\bowtie_{FK=K_2} \hat{R_2} = \joinView$. 
\end{proposition}

\paratitle{Complexity:} 
\revc{
The complexity of Algorithm \ref{algo:two_partition} is $O(n\cdot|S_{DC}|\cdot\binom{n}{T})$, where $|R_1| = n$, and $T$ is the number of tuples involved in the largest DC (assumed to be a constant), since the number of edges of each vertex can be at most $\binom{n}{T}$. 
The $n$ component stands for the possible need to iterate over all tuples in \joinView\ in line \ref{l:start1}, the $|S_{DC}|$ component stands for the possible need to iterate over all DCs when checking the condition in line \ref{l:tuple_subsets}, and the $\binom{n}{T}$ component is added due to the need to iterate over all subsets of $P_v$ that may satisfy a DC in lines \ref{l:tuple_subsets}--\ref{l:add_edge}. 
Since Algorithm \ref{algo:lf_coloring} (lines \ref{l:color}, \ref{l:color_s}) has a complexity of $O(n\cdot \binom{n}{T})$, and the loop (line \ref{l:add_row_R2}) has complexity of $O(n)$, they are not presented in the overall complexity of the algorithm. 
Note that $\bigcupdot_{j=1}^m P_v^j = \joinView$, where $m$ is the number of iterations and $P_v^j$ is the set $P_v^j$ generated in iteration $j$.
}
\\



\input{extension}

%% file: extension.tex
\paratitle{Extending the solution to snowflake schemas:}
\revb{
Our solution can be generalized to snowflake schemas in a manner similar to \cite{Arasu2011}. 
The idea is to start from the fact table (the central table) as $R_1$ and a table connected to it as $R_2$, i.e. going from the inside out in a Breadth-First Search manner. In every step, we include the previously completed tables in $R_1$, allowing CCs that span over the join view of multiple tables. 
This ensures that tuples are (possibly) added to the relation in the role of $R_2$ only once, since in the next step it would be considered as $R_1$ and thus maintain the foreign key dependency from the previous steps. 
}

\revb{
\begin{example}\label{ex:snowflake}
Consider a central Students table with two foreign key dependencies of a Majors table and a Courses table, and another foreign key connection to a Departments table through Majors:
\begin{center}
    \begin{tikzpicture}[scale=.7,level 1/.style={level distance=1.0cm, sibling distance=20mm}, level 2/.style={sibling distance=20mm}, level distance=1.0cm,
    every node/.style = {shape=rectangle, rounded corners, draw, align=center, top color=white, bottom color=blue!20},
    every edge/.style = {->},
    no edge/.style={
        edge from parent/.append style={draw=none}
    }]
        \node (r) {\footnotesize $^1$Students} [->]
        child { node (T1) {\footnotesize $^2$Majors} 
            child { node (T3) {\footnotesize $^4$Department} }
        }
        child { node (Tn) {\footnotesize $^3$Courses} };
    \end{tikzpicture}
\end{center}

The steps of the algorithm are as follows:
\begin{center}
\begin{footnotesize}
\begin{tabular}{| c | c | c | c |}
\hline 
Step & $R_1$ & $R_2$\\
\hline 
\hline 1 & $Students$ & $Majors$ \\
\hline 2 & $Students \bowtie Majors$ & $Courses$ \\
\hline 3 & $(Students \bowtie Majors) \bowtie Courses$ & $Departments$ \\
\hline
\end{tabular}
\end{footnotesize}
\end{center}
At each step we can therefore consider CCs over all tables we have considered so far. For example, in step 2, we can consider CCs over $((Students \bowtie Majors) \bowtie Courses)$ and not just over $Students \bowtie Courses$.  
Note that in the first step, we might have added artificial tuples to the Majors table. These are added without an $FK$ value that connects them to the Departments table so we account for them in the last step of connecting the Majors and Departments tables, making sure that the DCs that apply to the Majors table are satisfied. 
\end{example}
}

%% file: experiments.tex
\newcommand\nocell[1]{\multicolumn{#1}{c|}{}}
\newcommand\med{\cellcolor{gray!30}}
\newcommand\dark{\cellcolor{gray!50}}
\section{Experiments}
\label{sec:experiments}

We analyze the performance of our (hybrid) approach, and compare it with a baseline algorithm (based on \cite{Arasu2011}) in these terms:
\begin{compactenum}
    \itemsep0em
    \item Accuracy and runtime comparison between the baseline and our approach as data grows for fixed $S_{DC}$ and two settings of $S_{CC}$ (based on Section \ref{subsec:cc_relationships}): (i) $S_{CC}$ with no intersecting CCs, and (ii) $S_{CC}$ with intersecting CCs.
    
    \item 
    Accuracy and runtime comparison between the baseline and our approach for fixed data and combinations of good and bad $S_{DC}$ and $S_{CC}$. Good $S_{DC}$ creates zero cliques in conflict graphs and good $S_{CC}$ contains zero intersecting CCs.
    
    \item Runtime performance of our approach when data and $S_{DC}$ are kept fixed but the size of good $S_{CC}$ and bad $S_{CC}$ varies. 
    
    \item \revc{Runtime performance of our approach for fixed data, and good $S_{DC}$ and $S_{CC}$ as the number of columns in $R_2$ grows.}
    
\end{compactenum}

\reva{We implemented our solution and baseline in Python $3.6.9$ and Pandas DataFrame interface \cite{reback2020pandas}} on Tensor TXR$231$-$1000$R D$126$ Intel(R) Xeon(R) CPU E$5$-$2640$ v$4$ $2.40$GHz CPU with $512$ GB ($40$ cores) of RAM. We use the standard PuLP \cite{Mitchell11pulp:a} and NumPy Libraries for the ILP, and NetworkX \cite{HSSnetworkX2008} to construct and color conflict graphs.

\paratitle{A summary of our findings:}
\begin{compactenum}
    \itemsep0em
    \item Our approach satisfies all CCs in the absence of intersecting CCs with no error. Additionally, our approach satisfies all DCs (as guaranteed by our theoretical analysis), whereas the baseline does not (Figures~\ref{fig:exp1_varyData_err}-\ref{fig:exp2_4settings_error}). 
    Overall, our approach has the shortest runtime (Figure~\ref{fig:comparison_total_runtimes}) and achieves better accuracy for CCs and DCs together. Additionally, augmenting the input set of CCs with marginals over the non-key attributes in $R_1$ improves accuracy for CCs. 
    \reva{We also find that the time spent by the baseline on the ILP solver alone is comparable to the total time taken by our approach for larger data scales.}
    
    \item At a fixed data scale and for good and bad settings of DCs and CCs, where good DCs do not create cliques in conflict graphs and good CCs do not intersect, our approach has the shortest runtime. Its best performance is for good DCs and good CCs. In comparison, using bad DCs is slower because conflict graphs become denser, 
    and using bad CCs is even slower because of the ILP solver.
    
    \item Keeping the data and $S_{DC}$ fixed, we find that increasing the size of and/or intersections in $S_{CC}$ slows down \joinView\ completion, increasing the runtime of our approach (Figure~\ref{fig:exp3_breakdown_runtime_900}).
    
    \item \revc{We find that the time spent on coloring grows faster than that for recursing on Hasse diagrams as the number of columns in $R_2$ grows when good $S_{DC}$ and $S_{CC}$ are used.}
\end{compactenum}

\subsection{Setup}\label{subsec:setup}
We now describe the experimental setup (summarized in Table \ref{tbl:exp_settings_figs}) and define the error measures that are used to evaluate accuracy. We vary the database size (Table~\ref{tbl:data_scales}), DCs and CCs (Sec. \ref{sec:dcs_ccs}) 
to examine the scalability and accuracy of our solution. In Section~\ref{sec:exp_findings}, the errors and runtimes are averaged over $3$ independent runs.

\paratitle{Data.}
We perform experiments on a dataset that is derived from the 2010 U.S. Decennial Census \cite{sexton_abowd_schmutte_vilhuber_2017} comprising of two relations $Persons(p_{id}$, $Rel$, $Age$, \reva{$Multi$-$ling$}, $h_{id})$ and $Housing(h_{id}$, $Tenure$, $Area)$, with $Persons$ ($R_1$) missing all values in its foreign key column $h_{id}$. 
The different data scales are given in Table~\ref{tbl:data_scales}. 
By construction, \joinView\ and $Persons$ contain the same number of tuples. 
\revc{We also consider up to $10$ (non-key) columns in $Housing$, where we go from $(Tenure, Area)$ to $(Tenure, County, Area, St)$, add $(Div, Reg)$ and then add binary attributes $(Water, Bath)$ followed by $(Fridge, Stove)$. Note that values in $Div$ and $Reg$ are determined by the $St$ value.}
\vspace{-2mm}
\begin{table}[!ht]
    \centering \footnotesize
    \caption{Data scales given by the number of tuples}\label{tbl:data_scales}
    \begin{tabular}{| c | c | c | c | c |}
        \hline {\bf Scale} & 
        \begin{tabular}{@{}c@{}}{\bf Persons table}\end{tabular} & 
        \begin{tabular}{@{}c@{}}{\bf Housing table}\end{tabular} &
        \begin{tabular}{@{}c@{}}{\bf \joinView}\end{tabular}\\
        \hline $1\times$ & $25,099$ & $9,820$ & $25,099$\\
        \hline $2\times$ & $50,039$ & $19,640$ & $50,039$\\
        \hline \revb{$5\times$} & $124,746$ & $49,100$ & $124,746$\\
        \hline \revb{$10\times$} & $249,259$ & $98,200$ & $249,259$\\
        \hline \revb{$40\times$} & $1,015,686$ & $392,800$ & $1,015,686$\\
        \hline \revb{$80\times$} & \revb{$2,043,975$} & \revb{$785,600$} & \revb{$2,043,975$}\\
        \hline \revb{$120\times$} & \revb{$3,064,328$} & \revb{$1,178,400$} & \revb{$3,064,328$}\\
        \hline \revb{$160\times$} & \revb{$4,097,471$} & \revb{$1,571,200$} & \revb{$4,097,471$}\\
        \hline
    \end{tabular}
\end{table}
\vspace{-2mm}

\paratitle{Denial Constraints.}
$S^{all}_{DC}$ is the set of DCs (Sec. \ref{sec:dcs_ccs}) that not only gives the permissible age gap between a homeowner ($Rel=$ Owner) and other members in the same home, but also limits the number of homeowners, spouses and unmarried partners per home. $S^{good}_{DC}$ contains first $8$ DCs, none of which create cliques in conflict graphs.

\paratitle{Cardinality Constraints.}
We use two sets of CCs, $S^{good}_{CC}$ and $S^{bad}_{CC}$, with $1001$ CCs each (Sec. \ref{sec:dcs_ccs}). We assume that each input CC specifies a condition on an attribute from both $R_1$ and $R_2$.
$S^{bad}_{CC}$ contains CCs with intersecting $Age$ intervals, but $S^{good}_{CC}$ does not. 

\paratitle{Error Measures.}
We measure \textit{relative CC error} ${err}_i$ as $\frac{|\hat{c_i} - c_i|}{\max(10, c_i)}$, where $\hat{c_i}$ and $c_i$ are $CC_i$'s (in $S_{CC}$) counts in the solution and input. \revc{We use a threshold of $10$ in the denominator because some CCs have a target count of $0$ for small data scales.} We report the median relative CC errors in Figures~\ref{fig:exp1_varyData_err}-\ref{fig:exp2_4settings_error}, where $y=1$ represents $100\%$ error. 
We measure \textit{DC error} as the fraction of tuples in $\hat{R_1}$ that violate a $DC\in S_{DC}$. E.g., if $h_{id}$ value in the first two tuples in $Persons$ relation in Figure~\ref{fig:db_table_complete} was $2$, then the DC error would be $2/9$.

\paratitle{Baseline.}
Arasu et al. \cite{Arasu2011} focuses on the generation of synthetic databases with snowflake schema, where all joins are foreign key joins. This work considers CCs alone (no DCs) and imputes $FK$ using \joinView. Motivated by this work, we establish the two baseline versions given below (Section~\ref{sec:related} surveys more related works).

\begin{compactenum}
\item {\bf Baseline:} First, we use Algorithm~\ref{algo:cc_intersection} (without the for loop on line~\ref{l:augment_S_CC}) to fill-in tuples in \joinView. Any \joinView\ tuple without an assignment is completed by randomly assigning values in $B_1, \ldots, B_q$. In phase \rom{2}, we randomly assign a value from the candidate $FK$ values given by \joinView\ for each tuple in $R_1$.

\item {\bf Baseline with marginals:} We also study the impact of augmenting $S_{CC}$ with all $Age$--$Rel$--\reva{$Multi$-$ling$} (all-way) marginals from $Persons$, where domains of numerical attributes are broken 
using intervalization \cite{Arasu2011} 
on $S_{CC}$. 
Note that the marginals have equal target counts in $Persons$ and \joinView\ by construction. They ensure that each variable participates in 
the ILP, and is thus assigned a value in the solution. We find that this fills in all \joinView\ tuples. 
We refer to this algorithm as {\em baseline with marginals} that uses Algorithm~\ref{algo:cc_intersection} for phase \rom{1}, followed by random assignment in $FK$ using \joinView\ for phase \rom{2}. Hence, it falls in-between the baseline and our approach (Section~\ref{sec:hybrid_approach}). 
\end{compactenum}

\vspace{-2mm}
\begin{table}[h]
    \centering \footnotesize
    \caption{Datasets used in experiments, with details about the data scales, DCs and CCs given in Table~\ref{tbl:data_scales} and Section \ref{sec:dcs_ccs}}\label{tbl:inputs}
    \begin{tabular}{| c | c | c | c |}
        \hline \begin{tabular}{@{}c@{}}{\bf \revc{Dataset no.}}\end{tabular} &
        \begin{tabular}{@{}c@{}}{\bf Data Scale}\end{tabular} & 
        \begin{tabular}{@{}c@{}}{\bf DCs}\end{tabular} &
        \begin{tabular}{@{}c@{}}{\bf CCs}\end{tabular}\\
        \hline 
        $1$-$5$ &  \revb{$1\times$ to $40\times$} & $S^{all}_{DC}$ & $S^{good}_{CC}$\\
        \hline
        $6$-$10$ &  \revb{$1\times$ to $40\times$} & $S^{all}_{DC}$ & $S^{bad}_{CC}$\\
        \hline
        $11$ & \revb{$10\times$} & $S^{good}_{DC}$ & $S^{good}_{CC}$\\
        $12$ & \revb{$10\times$} & $S^{good}_{DC}$ & $S^{bad}_{CC}$\\
        \hline
        $13-17$ & \revb{$10\times$} & $S^{all}_{DC}$ & \begin{tabular}{@{}l@{}} ($500, 600, 700, 800,900$) $S^{good}_{CC}$\end{tabular}\\
        \hline
        $18-22$ & \revb{$10\times$} & $S^{all}_{DC}$ & \begin{tabular}{@{}l@{}} ($500, 600, 700, 800,900$) $S^{bad}_{CC}$\end{tabular}\\
        \hline
        \revb{$23-26$} & \revb{$40\times$ to $160\times$} & \revb{$S^{good}_{DC}$} & \revb{$S^{good}_{CC}$}\\
        \hline
        \revb{$27-30$} & \revb{$40\times$ to $160\times$} & \revb{$S^{good}_{DC}$} & \revb{$S^{bad}_{CC}$}\\
        \hline
        \revc{$31-34$} & \begin{tabular}{@{}l@{}} \revb{$10\times$} \revc{($4, 6, 8, 10$ non-}\\ \revc{key $R_2$ columns)} \end{tabular} & \revc{$S^{good}_{DC}$} & \revc{$S^{good}_{CC}$}\\
        \hline
    \end{tabular}
\end{table}
\vspace{-5mm}

\begin{table}[h]
    \centering \footnotesize
    \caption{Experimental settings for Figures~\ref{fig:exp1_varyData_err}-\ref{fig:exp3_breakdown_runtime_900}. Table~\ref{tbl:inputs} contains details about the input datasets}\label{tbl:exp_settings_figs}
    \begin{tabular}{| c | c | c | c |}
        \rowcolor[HTML]{C0C0C0} 
        \hline
        {\bf Experiment} & 
        {\bf Figure} & 
        {\bf Algorithm} &
        {\bf Input datasets}\\
        \hline
        \hline  \cellcolor[HTML]{D7D7D7}& \ref{tbl:exp1_varyScale_good_err} & Baselines vs Hybrid & $1$-$5$ \\
        \cline{2-4} \cellcolor[HTML]{D7D7D7}& \ref{tbl:exp1_varyScale_bad_err} & Baselines vs Hybrid & $6$-$10$ \\
        \cline{2-4} \cellcolor[HTML]{D7D7D7}& \ref{fig:exp1_box} & Baselines vs Hybrid & $10$\\
        \cline{2-4} \multirow{-4}{*}{\cellcolor[HTML]{D7D7D7} Accuracy Exp.} & \begin{tabular}{@{}c@{}}\ref{fig:exp2_4settings_error}\end{tabular} & Baselines vs Hybrid & $11, 12, 4, 9$ \\
        \hline
        \hline \cellcolor[HTML]{D7D7D7} & \ref{fig:comparison_total_runtimes} & Baselines vs Hybrid & $9, 10$ \\
        \cline{2-4} \cellcolor[HTML]{D7D7D7}& \revb{\ref{fig:larger_scales_gC_bC}} & \revb{Hybrid} & \revb{$11, 23-26$, $12$, $27-30$}\\
        \cline{2-4} \cellcolor[HTML]{D7D7D7}& \revc{\ref{fig:new_scaleR2_2_to_10_total}} & \revc{Hybrid} & \revc{$11, 31-34$}\\
        \cline{2-4} \multirow{-4}{*}{\cellcolor[HTML]{D7D7D7} Scalability Exp.}& \ref{fig:exp3_breakdown_runtime_900} & Hybrid & $17, 22$\\
        \hline
    \end{tabular}
    \vspace{-4mm}
\end{table}

\begin{figure*}[!htb]
    \footnotesize
    \captionsetup{justification=centering}
    \centering
    \begin{subfigure}[b]{0.49\textwidth}
        \begin{center}
            \begin{tabular}{|>{\centering\arraybackslash}p{0.5cm}|>{\centering\arraybackslash}p{0.9cm}|>{\centering\arraybackslash}p{1.1cm}|>{\centering\arraybackslash}p{0.73cm}|>{\centering\arraybackslash}p{0.9cm}|>{\centering\arraybackslash}p{1.1cm}|>{\centering\arraybackslash}p{0.73cm}|}
                \hline
                \cellcolor{blue!25} &
                \multicolumn{3}{c|}{\cellcolor{blue!25}\bf CC error}&
                \multicolumn{3}{c|}{\cellcolor{blue!25}\bf DC error}\\
                \cline{2-7}
                \multirow{-2}{*}{\cellcolor{blue!25}{\bf Scale}} & 
                \cellcolor{blue!25}\begin{tabular}{@{}c@{}}{\bf Baseline}\end{tabular} &
                \cellcolor{blue!25}\begin{tabular}{@{}c@{}}{\bf Baseline}\\(marginals)\end{tabular} &
                \cellcolor{blue!25}{\bf Hybrid} &
                \cellcolor{blue!25}\begin{tabular}{@{}c@{}}{\bf Baseline}\end{tabular} &
                \cellcolor{blue!25}\begin{tabular}{@{}c@{}}{\bf Baseline}\\(marginals)\end{tabular} & 
                \cellcolor{blue!25}{\bf Hybrid}\\
                \hline
                \cellcolor{blue!25}$1\times$ & \dark$0.300$ & $0$ & $0$ & \med$0.218$ & \dark$0.445$ & $0$ \\
                \hline
                \cellcolor{blue!25}$2\times$ & \dark$0.367$ & $0$ & $0$ & \med$0.245$ & \dark$0.465$ & $0$ \\
                \hline
                \cellcolor{blue!25}$5\times$ & \dark$0.526$ & $0$ & $0$ & \med$0.274$ & \dark$0.446$ & $0$ \\
                \hline
                \cellcolor{blue!25}$10\times$ & \dark$0.604$ & $0$ & $0$ & \med$0.303$ & \dark$0.489$ & $0$ \\
                \hline
                \cellcolor{blue!25}$40\times$ & \dark$0.559$ & $0$ & $0$ & \med$0.371$ & \dark$0.520$ & $0$ \\
                \hline
            \end{tabular}
        \end{center}
        \caption{$S^{all}_{DC}$, $S^{good}_{CC}$}
        \label{tbl:exp1_varyScale_good_err}
    \end{subfigure}
    \begin{subfigure}[b]{0.49\textwidth}
        \begin{center}
            \begin{tabular}{|>{\centering\arraybackslash}p{0.5cm}|>{\centering\arraybackslash}p{0.9cm}|>{\centering\arraybackslash}p{1.1cm}|>{\centering\arraybackslash}p{0.73cm}|>{\centering\arraybackslash}p{0.9cm}|>{\centering\arraybackslash}p{1.1cm}|>{\centering\arraybackslash}p{0.73cm}|}
                \hline
                \cellcolor{blue!25} &
                \multicolumn{3}{c|}{\cellcolor{blue!25}\bf CC error}&
                \multicolumn{3}{c|}{\cellcolor{blue!25}\bf DC error}\\
                \cline{2-7}
                \multirow{-2}{*}{\cellcolor{blue!25}{\bf Scale}} & 
                \cellcolor{blue!25}\begin{tabular}{@{}c@{}}{\bf Baseline}\end{tabular} &
                \cellcolor{blue!25}\begin{tabular}{@{}c@{}}{\bf Baseline}\\(marginals)\end{tabular} &
                \cellcolor{blue!25}{\bf Hybrid} &
                \cellcolor{blue!25}\begin{tabular}{@{}c@{}}{\bf Baseline}\end{tabular} &
                \cellcolor{blue!25}\begin{tabular}{@{}c@{}}{\bf Baseline}\\(marginals)\end{tabular} & 
                \cellcolor{blue!25}{\bf Hybrid}\\
                \hline
                \cellcolor{blue!25}$1\times$ & \dark$0.233$ & $0$ & $0$ & \med$0.228$ & \dark$0.435$ & $0$ \\
                \hline
                \cellcolor{blue!25}$2\times$ & \dark$0.300$ & $0$ & $0$ & \med$0.246$ & \dark$0.434$ & $0$ \\
                \hline
                \cellcolor{blue!25}$5\times$ & \dark$0.467$ & $0$ & $0$ & \med$0.279$ & \dark$0.402$ & $0$ \\
                \hline
                \cellcolor{blue!25}$10\times$ & \dark$0.537$ & $0$ & $0$ & \med$0.305$ & \dark$0.510$ & $0$ \\
                \hline
                \cellcolor{blue!25}$40\times$ & \dark$0.580$ & $0$ & $0$ & \med$0.373$ & \dark$0.489$ & $0$ \\
                \hline
            \end{tabular}
        \end{center}
        \caption{$S^{all}_{DC}$, $S^{bad}_{CC}$}
        \label{tbl:exp1_varyScale_bad_err}
    \end{subfigure}
    \caption{Error rate comparison between the baseline, baseline with marginals and hybrid as data grows from \revb{Scale $1\times$--$40\times$} and: (a) $S^{all}_{DC}$ ($12$ DCs) and $S^{good}_{CC}$ ($1001$ CCs) are used, and (b) $S^{all}_{DC}$ ($12$ DCs) and $S^{bad}_{CC}$ ($1001$ CCs) are used}
    \label{fig:exp1_varyData_err}
    \vspace{-4mm}
\end{figure*}

\begin{figure}[!ht]
    \centering
    \includegraphics[width=0.8\linewidth, height=3.5cm]{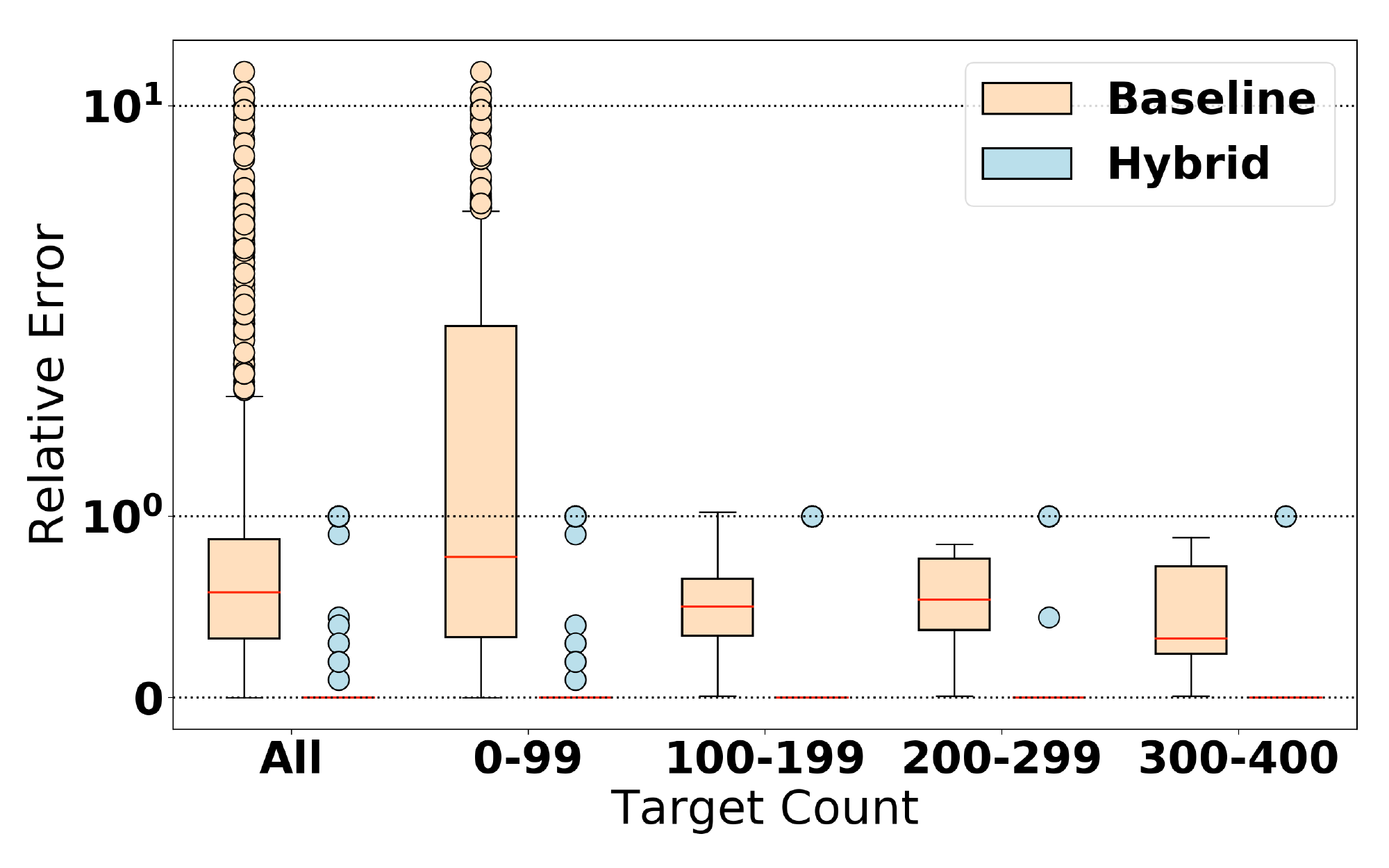}
    \caption{Relative CC error incurred by the baseline and hybrid for data \revb{Scale $40\times$}, $S^{all}_{DC}$ ($12$ DCs) and $S^{bad}_{CC}$ ($1001$ CCs). 
    We omit baseline with marginals as it satisfies all CCs 
    }
    \label{fig:exp1_box}
    \vspace{-4mm}
\end{figure}

\begin{figure}[!ht]
    \begin{center}
        \footnotesize
        \begin{tabular}{|>{\centering\arraybackslash}p{0.5cm}|>{\centering\arraybackslash}p{0.95cm}|>{\centering\arraybackslash}p{0.95cm}|>{\centering\arraybackslash}p{0.75cm}|>{\centering\arraybackslash}p{0.95cm}|>{\centering\arraybackslash}p{0.95cm}|>{\centering\arraybackslash}p{0.75cm}|}
            \hline
            \cellcolor{blue!25} &
            \multicolumn{3}{c|}{\cellcolor{blue!25}\bf CC error}&
            \multicolumn{3}{c|}{\cellcolor{blue!25}\bf DC error}\\
            \cline{2-7}
            \multirow{-2}{*}{\cellcolor{blue!25}\begin{tabular}{@{}c@{}}{\bf Data-}\\{\bf set}\end{tabular}} & 
            \cellcolor{blue!25}\begin{tabular}{@{}c@{}}{\bf Baseline}\\(no aug)\end{tabular} &
            \cellcolor{blue!25}\begin{tabular}{@{}c@{}}{\bf Baseline}\\(with aug)\end{tabular} &
            \cellcolor{blue!25}{\bf Hybrid} &
            \cellcolor{blue!25}\begin{tabular}{@{}c@{}}{\bf Baseline}\\(no aug)\end{tabular} &
            \cellcolor{blue!25}\begin{tabular}{@{}c@{}}{\bf Baseline}\\(with aug)\end{tabular} & 
            \cellcolor{blue!25}{\bf Hybrid}\\
            \hline
            \cellcolor{blue!25}11 & \dark$0.618$ & $0$ & $0$ & \dark$0.081$ & \med$0.009$ & $0$ \\
            \hline
            \cellcolor{blue!25}12 & \dark$0.573$ & $0$ & $0$ & \dark$0.079$ & \med$0.004$ & $0$ \\
            \hline
            \cellcolor{blue!25}4 & \dark$0.604$ & $0$ & $0$ & \med$0.303$ & \dark$0.489$ & $0$ \\
            \hline
            \cellcolor{blue!25}9 & \dark$0.537$ & $0$ & $0$ & \med$0.305$ & \dark$0.510$ & $0$ \\
            \hline
        \end{tabular}
    \end{center}
    \caption{CC and DC error comparison between baseline, baseline with marginals and hybrid for combinations of good and bad cases of DCs and CCs at data \revb{Scale $10\times$}}
    \label{fig:exp2_4settings_error}
    \vspace{-4mm}
\end{figure}

\begin{figure}[!ht]
    \centering
    \begin{subfigure}[b]{0.49\columnwidth}
        \centering
        \includegraphics[width=\linewidth]{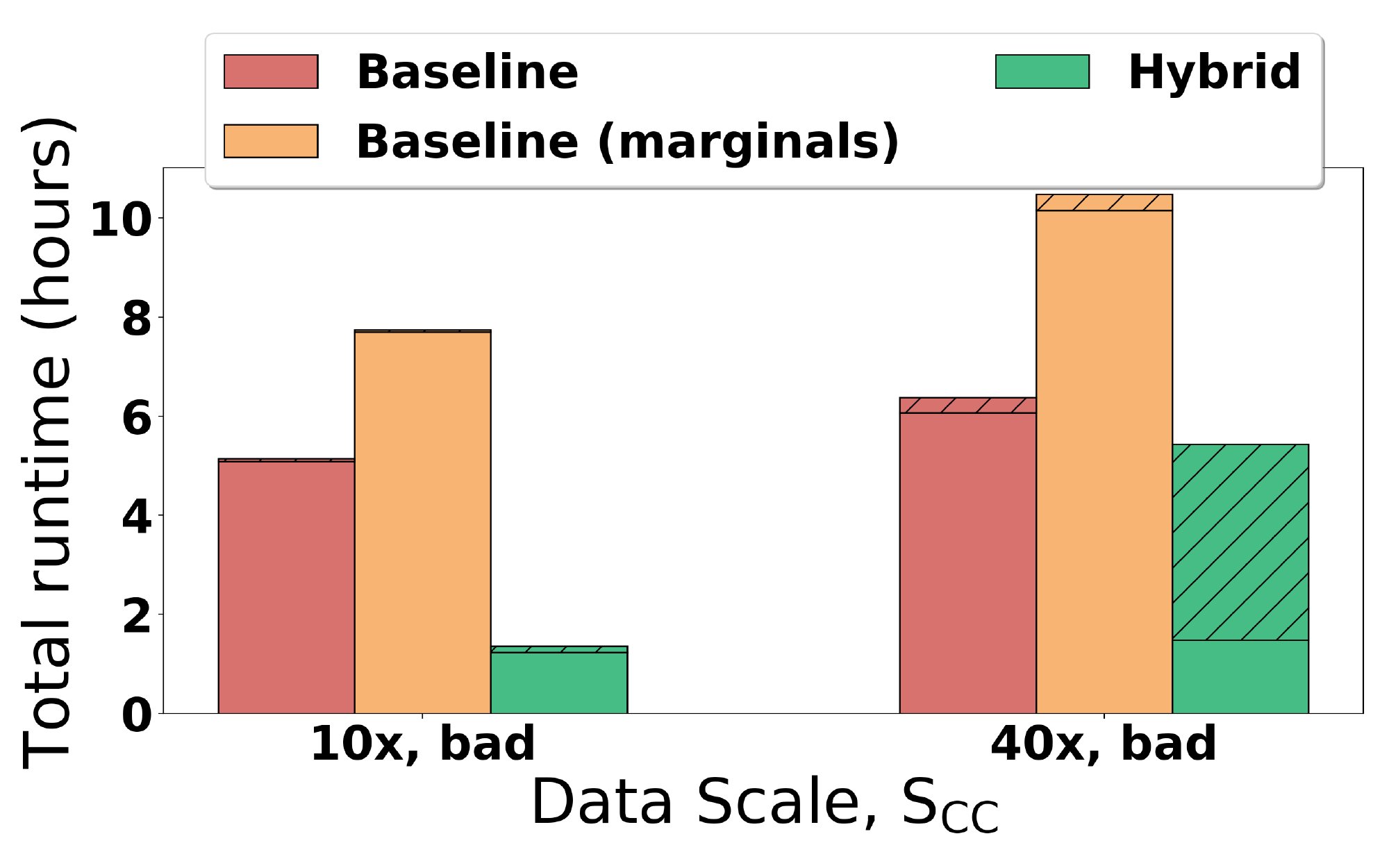}
        \caption{$S^{all}_{DC}$ ($12$ DCs), $S^{bad}_{CC}$}
        \label{fig:comparison_total_runtimes}
    \end{subfigure}
    \begin{subfigure}[b]{0.49\columnwidth}
        \centering
        \includegraphics[width=\linewidth]{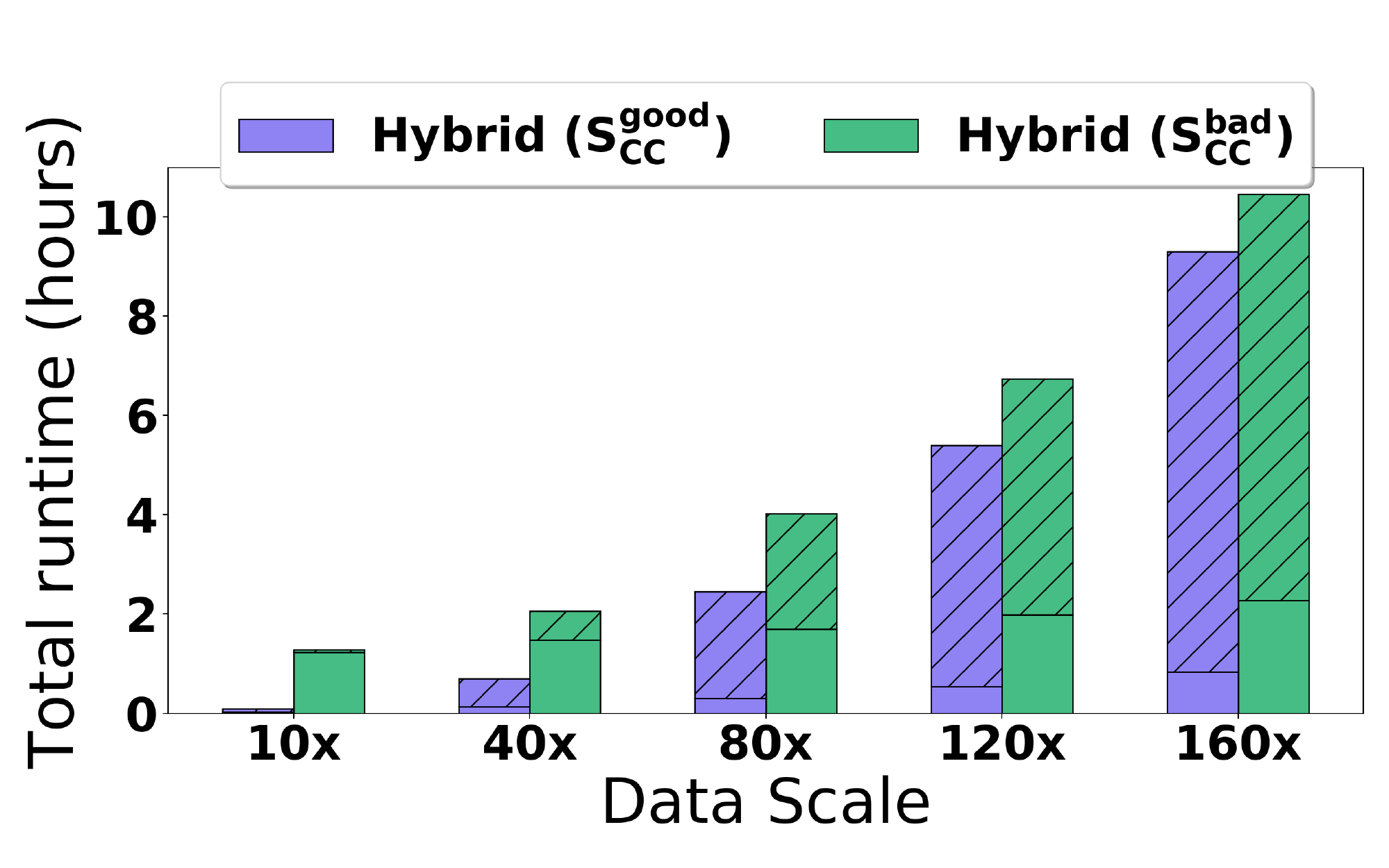}
        \caption{$S^{good}_{DC}$ ($8$ DCs), $S^{good}_{CC}$ or $S^{bad}_{CC}$}
        \label{fig:larger_scales_gC_bC}
    \end{subfigure}
    \caption{Shaded area depicts phase \rom{2} in: (a) Runtime comparison between baseline and hybrid for $S^{all}_{DC}$, $S^{bad}_{CC}$ ($1001$ CCs) and data \revb{Scale $10\times$ or $40\times$}, (b) \revb{Runtime of hybrid for $S^{good}_{DC}$ and data Scales $10\times$--$160\times$ with $S^{good}_{CC}$ or $S^{bad}_{CC}$ ($1001$ CCs)}}
    \label{fig:runtimes}
    \vspace{-4mm}
\end{figure}

\begin{figure}[!ht]
    \centering
    \includegraphics[width=0.8\linewidth, height=3.5cm]{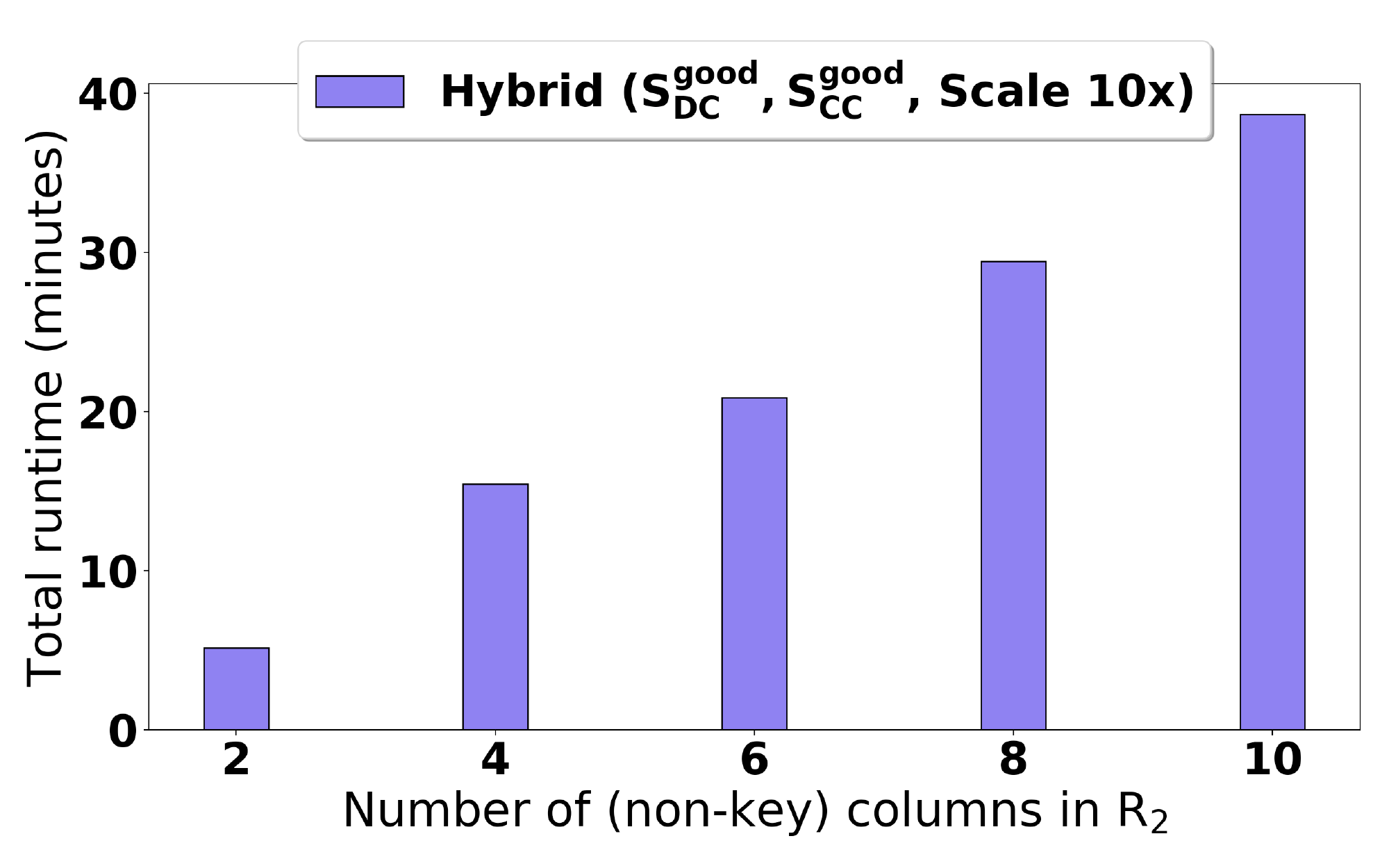}
    \caption{\revc{Runtime of hybrid for $S^{good}_{DC}$ ($8$ DCs), $S^{good}_{CC}$ ($1001$ CCs) and data Scale $10\times$ as number of columns in $R_2$ grows}}
    \label{fig:new_scaleR2_2_to_10_total}
    \vspace{-4mm}
\end{figure}

\subsection{Experimental Findings}
\label{sec:exp_findings}
We discuss results and address aspects raised at the start of Section~\ref{sec:experiments}.

\paratitle{Our approach vs baselines - Accuracy.}
We consider the experimental setup from Table~\ref{tbl:exp_settings_figs} for the accuracy experiments, and detail our results in Figures~\ref{fig:exp1_varyData_err}-\ref{fig:exp2_4settings_error}. 
Our approach always satisfies all DCs, and all CCs in $S^{good}_{CC}$. For $S^{bad}_{CC}$ (Table~\ref{tbl:exp1_varyScale_bad_err}), the median CC error is $0$ but the smallest and largest average errors are $0.048$ and $0.093$ due to limitations in augmenting $S_{CC}$ (Section~\ref{sec:hybrid_approach}). 
In contrast, the baseline gives median CC and DC errors between $0.233$-$0.580$ and $0.228$-$0.373$, whereas baseline with marginals satisfies all CCs but gives DC errors between $0.402$-$0.510$. We take a closer look at the relative CC errors for data \revb{Scale $40\times$} and $S^{bad}_{CC}$ 
in Figure~\ref{fig:exp1_box}. 
Note that DCs are used only after \joinView\ is partitioned by $B_1, \ldots, B_q$ values, so CCs affect the quality of the solution given by Algorithm~\ref{algo:two_partition}. 

Next, we look at combinations of good and bad cases of DCs and CCs for data at \revb{Scale $10\times$}. Again, our approach satisfies all DCs and gives a median CC error of $0$. Here, half the CCs were passed into Algorithm~\ref{algo:cc_no_intersection} that satisfies CCs exactly. The remaining CCs are augmented 
(Section~\ref{sec:hybrid_approach}) to improve accuracy for Algorithm~\ref{algo:cc_intersection}. We find that most CCs have a relative error of $0$. 
For $S^{bad}_{CC}$, the average CC error given by our approach is $0.0735$. In contrast, baseline gives CC errors between $0.537$-$0.618$ and DC errors between $0.079$-$0.305$, whereas baseline with marginals satisfies all CCs but gives DC errors between $0.004$-$0.510$ due to random assignment in $R_1.FK$.

\paratitle{Our approach vs baselines - Runtime.}
We consider the experimental setup as given in Table~\ref{tbl:exp_settings_figs} for the scalability experiments. 
The total runtimes \revb{at data Scales $10\times$ and $40\times$} are given in Figure~\ref{fig:comparison_total_runtimes}. 
\reva{Observe that the time spent on phase \rom{2} by the baseline is minimal because it randomly assigns $FK$ values, whereas our approach colors conflict graphs to satisfy all DC exactly. 
In our approach, Algorithm \ref{algo:cc_intersection} and \ref{algo:two_partition} are the bottlenecks taking $25.23$ \% and $72.74$ \% of the total runtime for $S^{all}_{DC}$ and $S^{bad}_{CC}$ at data Scale $40\times$.}


At data \revb{Scale $40\times$}, the runtimes for completing \joinView\ for $S^{good}_{CC}$ and $S^{bad}_{CC}$ are as follows: (1) baseline takes \reva{$5.88$ hours and $6.07$ hours}, (2) baseline with marginals takes \reva{$9.75$ hours and $10.15$ hours}, and (3) our approach takes \reva{$7.79$ minutes and $1.48$ hours}. The corresponding total runtimes are: (1) \reva{$6.19$ hours and $6.38$ hours}, (2) \reva{$10.04$ hours and $10.49$ hours}, and (3) \reva{$1.06$ hours and $5.43$ hours} hours. 
Our approach has the shortest runtime in completing \joinView\ because we take advantage of the relationships between the CCs to separate out intersecting CCs from $S_{CC}$ which reduces the time to solve the ILP. In contrast, the baseline creates one large ILP with all CCs (with or without all-way marginals). 
In addition, our approach does not need the ILP solver for $S^{good}_{CC}$, further improving the runtime.

Note that the baselines do not take into account the DCs and randomly assign $FK$ values based on filled-in \joinView, so solving the ILP dominates the time to populate $h_{id}$ in $Persons$ ($R_1$). 
Our approach has the shortest runtime for ($S^{good}_{DC}, S^{good}_{CC}$) (\reva{at $5.17$ minutes}), followed by ($S^{all}_{DC}, S^{good}_{CC}$), ($S^{good}_{DC}, S^{bad}_{CC}$) and ($S^{all}_{DC}, S^{bad}_{CC}$) (\reva{at $1.36$ hours}). Intuitively, for fixed data, completing \joinView\ is faster for $S^{good}_{CC}$ and conflict graphs are more likely to have fewer edges for $S^{good}_{DC}$. For ($S^{good}_{DC}, S^{good}_{CC}$) and ($S^{all}_{DC}, S^{bad}_{CC}$): (1) baseline took \reva{$4.84$-$5.14$ hours}, and (2) baseline with marginals took \reva{close to $8$ hours}. Baseline ran faster because it runs the ILP solver without the marginals.

\paratitle{Larger data scales - Runtime.}
\revb{
We examine our solution's runtime for larger data scales when $S^{good}_{DC}$ is used with $S^{good}_{CC}$ vs $S^{bad}_{CC}$ (see Figure~\ref{fig:larger_scales_gC_bC}). We find that our solution scales well, taking a total of $9.3$ hours for $S^{good}_{CC}$ and $10.46$ hours for $S^{bad}_{CC}$ at data Scale $160\times$.}

\paratitle{Increasing the number of $R_2$ columns - Runtime.}
\revc{
We study the effect of increasing the number of $R_2$ columns on the runtime of our approach for $S^{good}_{DC}$, $S^{good}_{CC}$ and data Scale $10\times$ (see Figure~\ref{fig:new_scaleR2_2_to_10_total}). We describe in Section~\ref{subsec:setup} how the number of columns in $R_2$ grows from $2$ to $10$. Our approach takes a total of $5.17$ minutes for $2$ columns and $38.66$ minutes for $10$ columns because we only consider the columns that are used in $S_{CC}$ for $combo_{unused}$ (Algorithm~\ref{algo:cc_no_intersection}).
}

\paratitle{Increasing the number of CCs - Runtime and accuracy.}
Here we study the effect of the size of $S_{CC}$ on the runtime and error in our approach (the last row in Table~\ref{tbl:exp_settings_figs}). 
The breakdown of runtimes for $900$ CCs chosen from $S^{good}_{CC}$ and $S^{bad}_{CC}$ is given in Figure~\ref{fig:exp3_breakdown_runtime_900}.


\begin{figure}[!ht]
\vspace{-3mm}
    \begin{center}
        \footnotesize
        \begin{tabular}{|>{\centering\arraybackslash}p{1.4cm}|>{\centering\arraybackslash}p{1.5cm}|>{\centering\arraybackslash}p{0.7cm}|>{\centering\arraybackslash}p{1.5cm}|>{\centering\arraybackslash}p{0.7cm}|}
            \hhline{~|-|-|-|-|}
            \nocell{1} &
            \multicolumn{2}{c|}{\cellcolor{blue!25}\bf $900$ CCs from $S^{good}_{CC}$}&
            \multicolumn{2}{c|}{\cellcolor{blue!25}\bf $900$ CCs from $S^{bad}_{CC}$}\\
            \hhline{~|-|-|-|-|}
            \nocell{1} & 
            \cellcolor{blue!25}{\bf Time} &
            \cellcolor{blue!25}{\bf \%} &
            \cellcolor{blue!25}{\bf Time} &
            \cellcolor{blue!25}{\bf \%}\\
            \hline
            \cellcolor{blue!25}{\bf Pairwise Comparison} & \reva{$4.48$s} & \reva{$1.12$} & \reva{$4.24$s} & \reva{$0.10$} \\
            \hline
            \cellcolor{blue!25}{\bf Recursion} & \reva{$1.70$m} & \reva{$25.64$} & \reva{$1.29$m} & \reva{$1.76$} \\
            \hline
            \cellcolor{blue!25}{\bf ILP solver} & $-$ & $-$ & \reva{$1.06$h} & \reva{$86.21$} \\
            \hline
            \cellcolor{blue!25}{\bf Coloring} & \reva{$4.87$m} & \reva{$73.24$} & \reva{$8.77$m} & \reva{$11.93$} \\
            \hline
        \end{tabular}
    \end{center}
    \caption{Runtime breakdown of the hybrid approach for \revb{data Scale $10\times$} with $S^{all}_{DC}$ ($12$ DCs) and $900$ CCs from $S^{good}_{CC}$ or $S^{bad}_{CC}$ (overall we have $1001$ CCs for both)
    }
    \label{fig:exp3_breakdown_runtime_900}
    \vspace{-3mm}
\end{figure}

As more CCs are used, the time spent on labeling pairs of CCs as disjoint, contained or intersecting increases. Since $S^{good}_{CC}$ contains no intersecting CCs, the ILP solver is not used and the runtime is faster. Algorithm~\ref{algo:cc_no_intersection} takes \reva{$1.42$ minutes} for $500$ CCs and \reva{$1.78$ minutes} for $900$ CCs. More CCs not only cause more updates to \joinView, but may also add CCs where $Area$ is used without $Tenure$, creating tuples with a partial assignment in \joinView\ that are completed in line \ref{l:complete_partial_assignment} of the algorithm. For $900$ CCs, the total runtime is \reva{$6.65$ minutes}, of which \reva{$4.87$ minutes} are spent in filling-in $R_1$ (Algorithm~\ref{algo:two_partition}).

When more CCs are chosen from $S^{bad}_{CC}$, we see an increase in the number of CCs passed to Algorithm~\ref{algo:cc_intersection} that makes the ILP solver slower. Algorithm \ref{algo:cc_no_intersection} takes \reva{$1.21$ minutes} for $500$ CCs and \reva{$1.36$ minutes} for $900$ CCs, whereas Algorithm~\ref{algo:cc_intersection} takes $25.99$ minutes for $500$ CCs and $1.06$ hours for $900$ CCs. For $900$ CCs, the total runtime is \reva{$1.23$ hours}, of which \reva{$8.77$ minutes} are spent in completing $R_1$.

Our approach satisfies all DCs, and CCs in $S^{good}_{CC}$. The median and average CC error rates are $0$ and $0.034$-$0.092$, resp.
We propose an optimization for the (bottleneck) coloring step in Section \ref{sec:parallel}.

%% file: related.tex
\section{related work}\label{sec:related}
Data generation has been the focus of multiple works, e.g., \cite{MannilaR89,GraySEBW94,BrunoC05,HoukjaerTW06,BinnigKLO07,LoCH10,Arasu2011,TayDWSLL13,ShenA13,AliIAB13,BudaCMK13,RablDFSJ15,FazekasK18,SanghiSHT18}.
{\em The main novelty of this paper is the generation of foreign keys for existing database relations while reducing the error of a set of CCs and ensuring the satisfaction of a set of DCs that relate to the foreign key attribute.}

A prominent line of work uses CCs to define the desired parameters of the generated data \cite{BinnigKLO07,Arasu2011,SanghiSHT18}.
QAGen \cite{BinnigKLO07} was among the first system that focused on data generation in a \textit{query-aware} fashion. The target application was to test the performance of a database management system (DBMS) when given a database schema, one parametric Conjunctive Query and a collection of constraints on each operator. 
MyBenchmark \cite{LoCH10} extends \cite{BinnigKLO07} by generating a {\em set} of database instances that approximately satisfies the cardinality expectations from a set of query results. 
HYDRA \cite{SanghiSHT18} uses a \textit{declarative approach} that allows for 
the generation of a database summary that can be used for dynamically generating data for query execution. 
Arasu et. al. \cite{Arasu2011} proposed a framework that supports multiple CCs and generates data using a graphical model that converts the CCs to equations, using the concept of intervalization for efficient computations. Indeed, we have drawn on this work for Algorithm \ref{algo:cc_intersection}. 
These approaches allow for complex CCs, whereas our approach allows for DCs as well. 
A recent work \cite{SoltanaSB17} has proposed a solution for generating {\em multiple data samples} using a seed sample of the data (generated by previous work \cite{SoltanaSSB18}), statistical constraints and data validity constraints specified in OCL \cite{ocl}. 
UpSizeR \cite{TayDWSLL13} has focused on scaling the database 
while maintaining foreign key constraints. 
Data generation from the database schema and statistical information has also been studied \cite{ShenA13,RablDFSJ15}. 

The field of data privacy 
\cite{Sweene02,Winkler04a,Dwork06,MachanavajjhalaGKV06,LiMHMR15,McKennaMHM18,KotsogiannisTHF19} typically gives mechanisms that generate query answers that do not expose features of the underlying private data, rather than generate the data itself. 
Some works \cite{HayRMS10,ZhangCPSX17} focus on providing consistent query answers, but none, to our knowledge, consider queries over linked data that guarantee the satisfaction of a set of ICs. 
Yahalom et. al. \cite{YahalomSZ10} developed a framework for converting production data into test data by modeling it as a constrained satisfaction problem (CSP) using specific constraints that can be expressed as part of the CSP. 

Finally, DCs (without CCs) have been mainly explored in relation with data cleaning \cite{ChomickiM05,Afrati2009,RekatsinasCIR17,ChuIP13,FaginKK15,KolahiL09,BertossiKL13, GiladDR20}. 
Previous work on the subject has focused on two main approaches: (1) repairing attribute values in cells \cite{RekatsinasCIR17,ChuIP13,BertossiKL13} and (2) tuple deletion \cite{ChomickiM05,LopatenkoB07,GiladDR20}. 
\reva{
In this context, there has been previous work on automatically discovering DCs from the complete data \cite{discoverChuIP13,PenaAN19,LivshitsHIK20}. 
We consider DCs based on the FK column, which is missing. In many scenarios, as is the premise in many data cleaning works (e.g., \cite{ChomickiM05,ChuIP13,RekatsinasCIR17,GiladDR20}), such DCs can be naturally inferred from the schema or from domain knowledge. As in data cleaning, the constraints can be formulated by the users as logical statements \cite{RekatsinasCIR17} or as SQL queries \cite{GiladDR20}. 

}

%% file: conclusions.tex
\section{Conclusions and Limitations}
\label{sec:conc}
We have defined the problem of generating links between database relations using linear CCs and foreign key DCs, and proved that it is intractable. 
Therefore, we have shown a novel two-phase heuristic solution. Our solution first considers the CCs, with a hybrid approach that combines an ILP-based solution and a solution based on specific relationships between the CCs. 
Second, our approach utilizes a version of conflict graph coloring in order to find a completion of the tuples that satisfies all DCs. 
Our experimental results show that our solution is both accurate and scalable. 

There are many intriguing directions for future work. First, our solution focuses on linear CCs and a subset of DCs. Finding a solution when the constraints include non-linear CCs (e.g., CCs on the number of rows that share the same foreign key) and general DCs (e.g., DCs on tuples that do not share a foreign key) is an important extension of our approach. 
Second, in phase \rom{1}, we assume foreign key dependence that induces a one-to-one mapping between the tuples of $R_1$ and the tuples of \joinView. Examining other join dependencies that do not have this property is an interesting direction of exploration. 
Third, in phase \rom{2}, tuples may be artificially added to $R_2$ due to the coloring algorithm. Some scenarios may not allow such augmentation and thus require different solutions. 
\revb{Finally, the extension of our solution to non-relational databases, such as graph databases and wide-column store is another subject of future study.}


%% file: revision_appendix.tex
\appendix

\section{Appendix}
We next give more details about the paper, including the full proofs of all propositions, details about our experimental settings and an outline for an optimization of the coloring process. 

\subsection{Proofs}\label{sec:proofs}
We next detail the full proofs for all of the propositions in the paper. 

\begin{proof}[Proof of Proposition \ref{prop:hardness}]
We give a reduction from NAE-3SAT\footnote{See \url{https://en.wikipedia.org/wiki/Not-all-equal_3-satisfiability} for details} to \prob. 
In the NAE-3SAT problem, we are given a 3-CNF formula $\varphi$ and asked whether there is a satisfying assignment to $\varphi$ with every clause having at least one literal with the value False. 
Given a 3-CNF formula $\varphi = C_1 \land \ldots \land C_n$, where $x_1,\ldots, x_m$ are the propositional variables in $\varphi$, construct a relation $R_1(Var,\alpha,Cls,Chosen)$, where $Chosen$ is missing all values, and $Var,\alpha,Cls$ attributes take the following values:
\begin{enumerate}
    \itemsep0em
    \item $(x_i, 1, C_j, ?)$ if making $x_i$ \textit{True} makes $C_j$ \textit{True}
    \item $(x_i, 0, C_j, ?)$ if making $x_i$ \textit{False} makes $C_j$ \textit{True}
\end{enumerate}
We define $S_{DC}$ to be the set with the following two DCs:
\begin{enumerate}
    \itemsep0em
    \item $\forall t_1, t_2.~ \neg (t_1.Var = t_2.Var \land t_1.\alpha \neq t_2.\alpha \land  t_1.Chosen = t_2.Chosen)$ 
    \item $\forall t_1, t_2, t_3.~ \neg (t_1.Cls = t_2.Cls = t_3.Cls \land t_1.Chosen = t_2.Chosen = t_3.Chosen)$
\end{enumerate}
The goal is to complete the missing column $Chosen$ in $R_1$. 
CCs are not needed in the reduction. 
We define $R_2$ as containing two attributes: a primary key column $Chosen$, and another column $E$. 
$R_2$ contains the tuples $(0, a)$ and $(1, b)$, i.e., the domain for $Chosen$ is $\{0, 1\}$. 
DC (1) makes sure that if a tuple of the form $(x_i, 1, C_a)$ is chosen for the assignment, then $(x_i, 0, C_b)$ cannot be chosen as well and vice versa, and DC (2) enforces that for each clause, at least one literal in that clause will have the value True and at least one literal will have value False. 
We now show that a satisfying assignment for $\varphi$ exists iff there is a solution to \prob.

($\Rightarrow{}$) Assume there is a satisfying assignment $\alpha$ to $\varphi$ such that each clause contains a literal assigned to False. For each variable $x_i$ mapped to True, for tuples of the form $(x_i, 1, C_j, ?)$ and $(x_i, 0, C_l, ?)$ we fill-in $(x_i, 1, C_j, 1)$ and $(x_i, 0, C_l, 0)$, i.e., a tuple will be completed with a $Chosen$ value of $1$ iff $\alpha$ value is $1$. If $x_i$ was mapped to False, we do the opposite. 
We need to show that DCs (1) and (2) are satisfied. 
For DC (1), since the procedure described above gives only one of the tuples of the form $(x_i, 0, C_j, ?)$ the value $1$ and the other gets the value $0$ according to $\alpha$, this DC is satisfied. 
For DC (2), assume we have $(x_1, 1, C_j, ?)$, $(x_2, 1, C_j, ?)$, $(x_3, 0, C_j, ?)$ (w.l.o.g it could also be a different combination). Then, $\alpha$ cannot map the variables $x_1, x_2$ to True and $x_3$ to False. Hence, the $Chosen$ attribute of at least one of the above tuples will be $0$, satisfying DC (2). 
Thus, $\alpha$ defines a solution to \prob.

($\Leftarrow{}$) Assume there is a solution to \prob. Define the assignment $\alpha$ as follows. if $(x_i, 1, C_j, 1)$ then $\alpha(x_i) = True$ and otherwise, $\alpha(x_i) = False$. 
We prove that $\alpha$ is a proper assignment, i.e., each variable is mapped to either True or False and not both and we further prove that $\alpha$ satisfies $\varphi$ and each clause contains a literal mapped to False by $\alpha$. 
First, since the solution satisfies DC (1), two tuples containing the same variable name $x_i$ and different assignment values will be given different $Chosen$ values, i.e., $(x_i, 1, C_j, 1)$ iff $(x_i, 0, C_l, 0)$ and $(x_i, 0, C_j, 1)$ iff $(x_i, 1, C_l, 0)$, therefore, $\alpha$ is a proper assignment. 
Second, since the solution satisfies DC (2), for every three tuples of the form $t_1 = (x_1, 1, C_j, ?)$, $t_2 = (x_2, 1, C_j, ?)$, $t_3 = (x_3, 0, C_j, ?)$ (w.l.o.g.), 
it holds that at least one of the values $t_1[Chosen]$, $t_2[Chosen]$, $t_3[Chosen]$ is different than the rest, i.e., there is at least one $1$ value and one $0$ value in the tuples for $C_j$. 
The three tuples mentioned above correspond to the clause $C_j = (x_1 \lor x_2 \lor \neg x_3)$ (recall that in our reduction the second attribute suggests the assignment of $x_i$ that will satisfy the clause $C_j$, so since $x_1 = True$ will satisfy $C_j$, we initially have $(x_i, 1, C_j, ?)$).  Assume w.l.o.g., that $t_1.Chosen = 1$, and $t_2.Chosen = t_3.Chosen = 0$, then $C_j$ contains the literals $x_1 = True$, $x_2 = False$, and $x_3 = True$, thus having one literal mapped to True and one mapped to False. 
It is possible to verify that this mapping works for the other three combinations of literals as well ($(x_1 \lor x_2 \lor x_3)$, $(x_1 \lor \neg x_2 \lor \neg x_3)$, $(\neg x_1 \lor \neg x_2 \lor \neg x_3)$). 
This implies that $\alpha$ both satisfies $\varphi$ (since every clause contains at least one variable that satisfies it) and maps at least one variable in each clause to False. 
 \end{proof}

 \begin{proof}[proof of Proposition \ref{prop:hybrid}]
 
\begin{lemma}
\label{lemma:all_disjoint}
    If every pair of CCs in $S_{CC}$ is disjoint, Algorithm~\ref{algo:cc_no_intersection} fills-in \joinView\ in polynomial time and $\joinView\vDash\sigma,~\forall\sigma\in S_{CC}$.
\end{lemma}

\begin{proof}
Assume that every pair of CCs in $S_{CC}$ is disjoint. 
If $CC_i$ and $CC_j$ are disjoint, then (by definition) either the selection conditions on $R_1$ columns are disjoint or they are identical on $R_1$ columns but disjoint on $R_2$ columns. 
Let us examine a pair $CC_i$-$CC_j$ of CCs that is disjoint based on the first condition. Here, the set of \joinView\ tuples that can contribute towards the target count of $CC_i$ is disjoint from that of $CC_j$. Thus, satisfying one cannot affect the availability of tuples for the other. 
Another possibility is that $CC_i$ and $CC_j$ are disjoint based on the second condition. Here, the set of \joinView\ tuples that can contribute to both is the same but must contain at least as many tuples as the sum of the target counts for the two CCs since a \joinView\ exists that satisfies $S_{CC}$. 

Therefore, given that $S_{CC}$ contains only disjoint CCs, we go straight to line \ref{l:all_disjoint} and then to the loop on line \ref{l:assign_vals_loop}, where we first find as many tuples in \joinView\ as the target count for $CC_i$ and then assign the corresponding values from the selection condition in $CC_i$ to these tuples in line \ref{l:assign_vals_end}. Finally, if some tuples are missing some or all values from $R_2$ (the $B_1, \ldots, B_q$ values), combinations of values in $R_2$ columns that are not relevant to $S_{CC}$ are assigned to them in lines \ref{l:remaining}--\ref{l:complete_partial_assignment}. These do not add to the counts of the CCs since they use fresh combinations that do not appear in the selection conditions of the CCs in $S_{CC}$. 
\end{proof}

We prove the proposition by induction on the size of the diagrams $H\in \mathcal{H}$, $H = (V_H,E_H)$. 

If $|E_H| = 0$, the first case checked in Algorithm \ref{algo:cc_no_intersection} applies and by Lemma \ref{lemma:all_disjoint}, we are done. 

For the inductive hypothesis, assume the proposition holds true for all diagrams with size $< n$. 
We prove it for a diagram with size $n$. 
If the base case does not hold, the algorithm identifies the maximal element, $CC_m$, of $H$ in line \ref{l:maximal}. Then, for each one of $CC_m$'s children, $c$, it solves the problem for the sub-diagram whose maximal element is $c$. Based on the induction assumption, the completion of \joinView\ returned from this call satisfies all the CCs in these sub-diagrams (the children of $m$ are pairwise-disjoint). We first find the remaining number of tuples needed for $CC_m$ after the CCs of its children have been satisfied in line \ref{l:combine_with_root} and then in line \ref{l:combine} assign the $B_1, \ldots, B_q$ values given by $CC_m$. 
Finally, as shown in Lemma \ref{lemma:all_disjoint}, the values assigned to the tuples in lines \ref{l:remaining}--\ref{l:complete_partial_assignment} do not add to the count of any of the CCs. 
Thus, we find a solution that satisfies all of the CCs whose vertices are in $H$. 
\end{proof}

\begin{proof}[Proof of Proposition \ref{prop:color_satisfies}]
Given relation $R_1$ with an empty $FK$ column and the set $S_{DC}$ of FK DCs, consider the conflict hypergraph $G$. Suppose we have a proper coloring $c$ of $G$, and we assign each tuple $t\in R_1,~ t.FK = c(t)$. 

We now show that this satisfies the DCs. Consider a DC $\sigma = \forall t_1,\ldots, t_k.~ \neg (\varphi(t_1,\ldots, t_k) \land t_1.FK = \ldots = t_k.FK)$, and a set of tuples $t_1, \ldots, t_k$ such that $\varphi(t_1,\ldots, t_k) = True$. 
Therefore, there is an edge $\{t_1,$ $\ldots, t_k\}$ in $G$. Since $c$ is a proper coloring of $G$, at least one of the tuples $t_1, \ldots, t_k$ gets a different color than the rest. Suppose this tuple is $t_i$. Thus, $t_i.FK \neq t_j.FK$ for all $1\leq j\leq n, j\neq i$, and $t_1, \ldots, t_k$ does not violate $\sigma$. Since this is true for all such sets of tuples, $\sigma$ is satisfied by such an assignment of $FK$ values.
\end{proof}

\begin{proof}[Proof of Proposition \ref{prop:solution}]
First, we show that $\forall \sigma\in S_{DC},~ \hat{R_1} \vDash \sigma$. 
The algorithm starts by constructing the conflict hypergraph, $G_c$, in lines \ref{l:start1}--\ref{l:add_edge}, where each tuple in $\joinView$ with the same combination of $B_1, \ldots, B_q$ values is added as a vertex. Note that these vertices represent the tuples in $R_1$. 
Each set of vertices that violates the first part of an FK DC is added as an edge. 
According to Proposition \ref{prop:color_satisfies}, in line \ref{l:color}, the colored vertices represent tuples that now have FK values that satisfy the DCs that they are involved in. $s$ now contains all vertices that could not be colored. 
In lines \ref{l:new_colors_for_s}--\ref{l:add_row_R2}, the algorithm colors the vertices $t \in s$ using a fresh color $c[t]$ so all colored vertex combinations still do not violate any FK DC over $\hat{R_1}$.
Finally, in line \ref{l:handle_invalid} we construct a secondary conflict hypergraph and color the \noJoin\ vertices according to it, making sure that none of these tuples violate a DC over $\hat{R_1}$. 

We now prove that $\hat{R_1}\bowtie_{FK=K_2} \hat{R_2} = \joinView$. 
Prior to line \ref{l:add_row_R2}, we do not artificially add tuples to $\hat{R_2}$, so it is clear that $\hat{R_1}\bowtie_{FK=K_2} \hat{R_2} = \joinView$. 
Suppose we color $t \in \hat{R_1}$ with a new color $c[t]$ (i.e., $t.FK = c[t]$) in line \ref{l:color_s} and its corresponding tuple in \joinView\ is $t_J$. 
In the loop on line \ref{l:add_row_R2} we add a tuple, $t_{new}$, to $\hat{R_2}$ with $B_1, \ldots, B_q$ values as in $t_J$ and give it the primary key value $K_2 = c[t]$ (which is a fresh value not used in other tuples). 
Thus, joining $t$ with $t_{new}$ will give us $t_J$, which already exists in \joinView.
A similar procedure is performed in line \ref{l:handle_invalid} for the invalid vertices. Therefore, the added tuple does not add new tuples to \joinView, and we still have $\hat{R_1}\bowtie_{FK=K_2} \hat{R_2} = \joinView$.
\end{proof}

\subsection{Constraints Used in Section \ref{sec:experiments}}\label{sec:dcs_ccs}
We detail the CCs and DCs used in our experimental study. Table \ref{tbl:exp_dcs} depicts the DCs, and Table \ref{tbl:exp_ccs} depicts the CCs. 

\begin{table}[!ht]
    \centering \footnotesize
    \caption{List of all $12$ DCs from which subsets are used
    }\label{tbl:exp_dcs}
    \begin{tabular}{| c | l |}
        \hline {\bf \#} & {\bf DC ($A$ denotes the age of a homeowner)} \\
        \hline 1 & \begin{tabular}{@{}l@{}}No biological or adoptive or step child can have age outside $[A$-$69, A$-$12]$\\when homeowner \reva{is not multi-lingual}\end{tabular}\\ 
        \hline 2 & \begin{tabular}{@{}l@{}}No biological or adoptive or step child can have age outside $[A$-$50, A$-$12]$\\when homeowner \reva{is multi-lingual}\end{tabular}\\
        \hline 3 & No spouse or unmarried partner can have age outside $[A$-$50, A$+$50]$\\
        \hline 4 & No sibling can have age outside $[A$-$35, A$+$35]$\\
        \hline 5 & No parent or parent-in-law can have age outside $[A$+$12, A$+$115]$\\
        \hline 6 & No grandchild can have age outside $[A$-$115, A$-$30]$\\
        \hline 7 & No son or daughter in-law can have age outside $[A$-$69, A$-$1]$\\
        \hline 8 & No foster child can have age outside $[A$-$69, A$-$12]$\\
        \hline 9 & No two householders can share a house\\
        \hline 10 & \begin{tabular}{@{}l@{}}If $A<30$, number of grandchildren \& son/daughter in-law in the house\\must be $0$\end{tabular}\\
        \hline 11 & If $A>94$, number of parent/parent-in-law in the house must be $0$\\
        \hline 12 & No two spouses or unmarried partners can share a house\\
        \hline
    \end{tabular}
\end{table}

\begin{table}[!ht]
    \scriptsize
    \centering
    \caption{Set $S_{CC}^{good}$ and $S_{CC}^{bad}$ use subsets of predicates from the first and second table, respectively, with $469$ $Tenure$-$Area$ values and another $121$ $Area$ values without a $Tenure$ value.}
    \label{tbl:exp_ccs}
    
    \begin{tabular}{| c | c | c |}
        \multicolumn{3}{c}{$S_{CC}^{good}$}\\
        \hline {\bf Age} & {\bf Rel} & \begin{tabular}{@{}l@{}}{\bf \reva{Multi}}\\{\bf\reva{-ling}} \end{tabular} \\
        \hline $[18, 114]$ & Owner & 0 \\
        \hline $[18, 114]$ & Spouse & 1 \\
        \hline $[0, 10]$ & Biological child & $-$\\ 
        \hline $[6, 10]$ & Biological child & $-$\\ 
        \hline $[2, 5]$ & Biological child & $-$\\
        \hline $[3, 5]$ & Biological child & $-$\\
        \hline $[3, 5]$ & Biological child & 0 \\
        \hline $[11, 18]$ & Biological child & $-$\\  
        \hline $[11, 13]$ & Biological child & $-$\\
        \hline $[14, 18]$ & Biological child & $-$\\
        \hline $[19, 30]$ & Biological child & $-$\\
        \hline $[22, 30]$ & Biological child & $-$\\ 
        \hline $[25, 30]$ & Biological child & 1\\ 
        \hline $[18, 39]$ & Father/Mother & $-$\\
        \hline $[40, 85]$ & Father/Mother & 0 \\
        \hline $[40, 85]$ & Father/Mother & 1 \\
        \hline $[15, 85]$ & House/Room mate & 0 \\
        \hline $[15, 85]$ & House/Room mate & 1 \\
        \hline $[18, 30]$ & Grandchild & 0 \\
        \hline $[18, 30]$ & Grandchild & 1 \\
        \hline $[18, 114]$ & Unmarried partner & 1 \\
        \hline $[0, 30]$ & Step child & $-$\\
        \hline $[0, 20]$ & Step child & $-$\\
        \hline $[21, 30]$ & Step child & 1\\
        \hline $[19, 40]$ & Adopted child & $-$\\
        \hline $[25, 40]$ & Adopted child & 1 \\
        \hline $[31, 40]$ & Adopted child & 1 \\
        \hline
    \end{tabular}
    \quad
    \begin{tabular}{| c | c | c |}
        \multicolumn{3}{c}{$S_{CC}^{bad}$}\\
        \hline {\bf Age} & {\bf Rel} & \begin{tabular}{@{}l@{}}{\bf \reva{Multi}}\\{\bf\reva{-ling}} \end{tabular} \\
        \hline $[18, 114]$ & Owner & 0 \\
        \hline $[18, 114]$ & Spouse & 1 \\
        \hline $[0, 10]$ & Biological child & $-$\\ 
        \hline $[6, 10]$ & Biological child & $-$\\ 
        \hline $[2, 5]$ & Biological child & $-$\\
        \hline $[3, 5]$ & Biological child & 0 \\
        \hline $[11, 18]$ & Biological child & $-$\\  
        \hline $[11, 13]$ & Biological child & $-$\\
        \hline $[14, 18]$ & Biological child & $-$\\
        \hline $[19, 30]$ & Biological child & $-$\\
        \hline $[22, 30]$ & Biological child & $-$\\
        \hline $[40, 85]$ & Father/Mother & 0 \\
        \hline $[40, 85]$ & Father/Mother & 1 \\
        \hline $[15, 85]$ & House/Room mate & 0 \\
        \hline $[15, 85]$ & House/Room mate & 1 \\
        \hline $[18, 30]$ & Grandchild & 0 \\
        \hline $[18, 30]$ & Grandchild & 1 \\
        \hline $[18, 114]$ & Unmarried partner & 1 \\
        \hline $[0, 30]$ & Step child & $-$\\
        \hline $[21, 114]$ & Spouse & 1 \\
        \hline $[21, 64]$ & Spouse & 1 \\
        \hline $[18, 39]$ & Spouse & 1 \\
        \hline $[18, 85]$ & Spouse & 1 \\
        \hline $[40, 85]$ & Spouse & 1 \\
        \hline $[65, 114]$ & Father/Mother & 1 \\
        \hline $[0, 39]$ & Grandchild & 1 \\
        \hline $[22, 39]$ & Grandchild & 1 \\
        \hline $[0, 21]$ & Step child & $-$\\
        \hline $[19, 39]$ & Adopted child & $-$\\
        \hline $[25, 39]$ & Adopted child & 1 \\
        \hline $[31, 39]$ & Adopted child & 1 \\
        \hline
    \end{tabular}
\end{table}

\subsection{Parallelizing the Coloring Processes}\label{sec:parallel}
In Section \ref{sec:dc_algo}, we describe an optimization that allows for the conflict hypergraph to be split according to the filled-in \joinView\ and $R_2$ relations by the $B_1, \ldots, B_q$ values. As a result, we color each component of the hypergraph individually with Algorithm \ref{algo:lf_coloring} (in lines \ref{l:color} and \ref{l:color_s}) in different iterations of the loop in line \ref{l:start1}. 
For example, in our experiments with Dataset no. $11$ in Table \ref{tbl:inputs}, the conflict hypergraph was split into $3805-3813$ components. 
Thus, it is possible to parallelize the coloring process performed in each iteration of the loop in line \ref{l:start1} in Algorithm \ref{algo:two_partition} and color separate components obtained in different iterations on different machines. 

%% file: main.bbl
\begin{thebibliography}{10}

\bibitem{ocl}
Object constraint language 2.4 specification.
\newblock \url{https://www.omg.org/spec/OCL/About-OCL/}, 2017.

\bibitem{achlioptas_molloy_1997}
Dimitris Achlioptas and Michael Molloy.
\newblock The analysis of a list-coloring algorithm on a random graph (extended
  abstract.
\newblock {\em IEEE}, 1997.

\bibitem{Afrati2009}
Foto~N. Afrati and Phokion~G. Kolaitis.
\newblock Repair checking in inconsistent databases: Algorithms and complexity.
\newblock In {\em ICDT}, pages 31--41, 2009.

\bibitem{AliIAB13}
Shaukat Ali, Muhammad Zohaib~Z. Iqbal, Andrea Arcuri, and Lionel~C. Briand.
\newblock Generating test data from {OCL} constraints with search techniques.
\newblock {\em {IEEE} Trans. Software Eng.}, 39(10):1376--1402, 2013.

\bibitem{Arasu2011}
Arvind Arasu, Raghav Kaushik, and Jian Li.
\newblock Data generation using declarative constraints.
\newblock In {\em SIGMOD}, pages 685--696, 2011.

\bibitem{HSSnetworkX2008}
Daniel A.~Schult Aric A.~Hagberg and Pieter~J. Swart.
\newblock Exploring network structure, dynamics, and function using networkx.
\newblock SciPy, 2008.

\bibitem{BarakCDKMT07}
Boaz Barak, Kamalika Chaudhuri, Cynthia Dwork, Satyen Kale, Frank McSherry, and
  Kunal Talwar.
\newblock Privacy, accuracy, and consistency too: a holistic solution to
  contingency table release.
\newblock In {\em SIGACT-SIGMOD-SIGART}, pages 273--282, 2007.

\bibitem{BertossiKL13}
Leopoldo~E. Bertossi, Solmaz Kolahi, and Laks V.~S. Lakshmanan.
\newblock Data cleaning and query answering with matching dependencies and
  matching functions.
\newblock {\em Theory Comput. Syst.}, 52(3):441--482, 2013.

\bibitem{BinnigKLO07}
Carsten Binnig, Donald Kossmann, Eric Lo, and M.~Tamer {\"{O}}zsu.
\newblock Qagen: generating query-aware test databases.
\newblock In {\em SIGMOD}, pages 341--352, 2007.

\bibitem{bohannon2007conditional}
Philip Bohannon, Wenfei Fan, Floris Geerts, Xibei Jia, and Anastasios
  Kementsietsidis.
\newblock Conditional functional dependencies for data cleaning.
\newblock In {\em ICDE}, pages 746--755, 2007.

\bibitem{BrunoC05}
Nicolas Bruno and Surajit Chaudhuri.
\newblock Flexible database generators.
\newblock In {\em Proc. {VLDB} Endow.}, pages 1097--1107, 2005.

\bibitem{BudaCMK13}
Teodora~Sandra Buda, Thomas Cerqueus, John Murphy, and Morten Kristiansen.
\newblock {VFDS:} very fast database sampling system.
\newblock In {\em IRI}, pages 153--160, 2013.

\bibitem{10.1145/248603.248616}
Surajit Chaudhuri and Umeshwar Dayal.
\newblock An overview of data warehousing and olap technology.
\newblock {\em SIGMOD Rec.}, 26(1):65–74, March 1997.

\bibitem{ChomickiM05}
Jan Chomicki and Jerzy Marcinkowski.
\newblock Minimal-change integrity maintenance using tuple deletions.
\newblock {\em Inf. Comput.}, 197(1-2):90--121, 2005.

\bibitem{discoverChuIP13}
Xu~Chu, Ihab~F. Ilyas, and Paolo Papotti.
\newblock Discovering denial constraints.
\newblock {\em {PVLDB}}, 6(13):1498--1509, 2013.

\bibitem{ChuIP13}
Xu~Chu, Ihab~F. Ilyas, and Paolo Papotti.
\newblock Holistic data cleaning: Putting violations into context.
\newblock In {\em ICDE}, pages 458--469, 2013.

\bibitem{Dwork06}
Cynthia Dwork.
\newblock Differential privacy.
\newblock In {\em ICALP}, volume 4052, pages 1--12, 2006.

\bibitem{FaginKK15}
Ronald Fagin, Benny Kimelfeld, and Phokion~G. Kolaitis.
\newblock Dichotomies in the complexity of preferred repairs.
\newblock In {\em PODS}, pages 3--15, 2015.

\bibitem{FazekasK18}
B{\'{a}}lint Fazekas and Attila Kiss.
\newblock Statistical data generation using sample data.
\newblock In {\em New Trends in Databases and Information Systems}, volume 909,
  pages 29--36, 2018.

\bibitem{GiladDR20}
Amir Gilad, Daniel Deutch, and Sudeepa Roy.
\newblock On multiple semantics for declarative database repairs.
\newblock In {\em SIGMOD}, pages 817--831, 2020.

\bibitem{GraySEBW94}
Jim Gray, Prakash Sundaresan, Susanne Englert, Kenneth Baclawski, and Peter~J.
  Weinberger.
\newblock Quickly generating billion-record synthetic databases.
\newblock In {\em SIGMOD}, pages 243--252, 1994.

\bibitem{HayRMS10}
Michael Hay, Vibhor Rastogi, Gerome Miklau, and Dan Suciu.
\newblock Boosting the accuracy of differentially private histograms through
  consistency.
\newblock {\em Proc. {VLDB} Endow.}, 3(1):1021--1032, 2010.

\bibitem{HeMD14}
Xi~He, Ashwin Machanavajjhala, and Bolin Ding.
\newblock Blowfish privacy: tuning privacy-utility trade-offs using policies.
\newblock In {\em SIGMOD}, pages 1447--1458. {ACM}, 2014.

\bibitem{HoukjaerTW06}
Kenneth Houkj{\ae}r, Kristian Torp, and Rico Wind.
\newblock Simple and realistic data generation.
\newblock In {\em Proc. {VLDB} Endow.}, pages 1243--1246, 2006.

\bibitem{jensen2011graph}
Tommy~R Jensen and Bjarne Toft.
\newblock {\em Graph coloring problems}, volume~39.
\newblock John Wiley \& Sons, 2011.

\bibitem{KolahiL09}
Solmaz Kolahi and Laks V.~S. Lakshmanan.
\newblock On approximating optimum repairs for functional dependency
  violations.
\newblock In {\em ICDT}, pages 53--62, 2009.

\bibitem{KotsogiannisTHF19}
Ios Kotsogiannis, Yuchao Tao, Xi~He, Maryam Fanaeepour, Ashwin Machanavajjhala,
  Michael Hay, and Gerome Miklau.
\newblock Privatesql: {A} differentially private {SQL} query engine.
\newblock {\em Proc. {VLDB} Endow.}, 12(11):1371--1384, 2019.

\bibitem{LiHMW14}
Chao Li, Michael Hay, Gerome Miklau, and Yue Wang.
\newblock A data- and workload-aware query answering algorithm for range
  queries under differential privacy.
\newblock {\em Proc. {VLDB} Endow.}, 7(5):341--352, 2014.

\bibitem{LiMHMR15}
Chao Li, Gerome Miklau, Michael Hay, Andrew McGregor, and Vibhor Rastogi.
\newblock The matrix mechanism: optimizing linear counting queries under
  differential privacy.
\newblock {\em {VLDB} J.}, 24(6):757--781, 2015.

\bibitem{LivshitsHIK20}
Ester Livshits, Alireza Heidari, Ihab~F. Ilyas, and Benny Kimelfeld.
\newblock Approximate denial constraints.
\newblock {\em Proc. {VLDB} Endow.}, 13(10):1682--1695, 2020.

\bibitem{LoCH10}
Eric Lo, Nick Cheng, and Wing{-}Kai Hon.
\newblock Generating databases for query workloads.
\newblock {\em Proc. {VLDB} Endow.}, 3(1):848--859, 2010.

\bibitem{LopatenkoB07}
Andrei Lopatenko and Leopoldo~E. Bertossi.
\newblock Complexity of consistent query answering in databases under
  cardinality-based and incremental repair semantics.
\newblock In {\em ICDT}, pages 179--193, 2007.

\bibitem{MachanavajjhalaGKV06}
Ashwin Machanavajjhala, Daniel Kifer, Johannes Gehrke, and Muthuramakrishnan
  Venkitasubramaniam.
\newblock \emph{L}-diversity: Privacy beyond \emph{k}-anonymity.
\newblock {\em {ACM} Trans. Knowl. Discov. Data}, 1(1):3, 2007.

\bibitem{MannilaR89}
Heikki Mannila and Kari{-}Jouko R{\"{a}}ih{\"{a}}.
\newblock Automatic generation of test data for relational queries.
\newblock {\em J. Comput. Syst. Sci.}, 38(2):240--258, 1989.

\bibitem{McKennaMHM18}
Ryan McKenna, Gerome Miklau, Michael Hay, and Ashwin Machanavajjhala.
\newblock Optimizing error of high-dimensional statistical queries under
  differential privacy.
\newblock {\em Proc. {VLDB} Endow.}, 11(10):1206--1219, 2018.

\bibitem{Mckenna2019}
Ryan McKenna, Daniel Sheldon, and Gerome Miklau.
\newblock Graphical-model based estimation and inference for differential
  privacy.
\newblock In {\em ICML}, volume~97, pages 4435--4444, 2019.

\bibitem{Mitchell11pulp:a}
Stuart Mitchell, Stuart~Mitchell Consulting, and Iain Dunning.
\newblock Pulp: A linear programming toolkit for python.
\newblock \url{http://www.optimization-online.org/DB_FILE/2011/09/3178.pdf},
  2011.

\bibitem{reback2020pandas}
The pandas~development team.
\newblock pandas-dev/pandas: Pandas, February 2020.

\bibitem{PenaAN19}
Eduardo H.~M. Pena, Eduardo~Cunha de~Almeida, and Felix Naumann.
\newblock Discovery of approximate (and exact) denial constraints.
\newblock {\em Proc. {VLDB} Endow.}, 13(3):266--278, 2019.

\bibitem{tpchPaper}
Meikel P{\"{o}}ss and Chris Floyd.
\newblock New {TPC} benchmarks for decision support and web commerce.
\newblock {\em {SIGMOD} Rec.}, 29(4):64--71, 2000.

\bibitem{RablDFSJ15}
Tilmann Rabl, Manuel Danisch, Michael Frank, Sebastian Schindler, and
  Hans{-}Arno Jacobsen.
\newblock Just can't get enough: Synthesizing big data.
\newblock In {\em SIGMOD}, pages 1457--1462, 2015.

\bibitem{RekatsinasCIR17}
Theodoros Rekatsinas, Xu~Chu, Ihab~F. Ilyas, and Christopher R{\'{e}}.
\newblock Holoclean: Holistic data repairs with probabilistic inference.
\newblock {\em {PVLDB}}, 10(11):1190--1201, 2017.

\bibitem{SanghiSHT18}
Anupam Sanghi, Raghav Sood, Jayant~R. Haritsa, and Srikanta Tirthapura.
\newblock Scalable and dynamic regeneration of big data volumes.
\newblock In {\em EDBT}, pages 301--312, 2018.

\bibitem{sexton_abowd_schmutte_vilhuber_2017}
William Sexton, John~M. Abowd, Ian~M. Schmutte, and Lars Vilhuber.
\newblock Synthetic population housing and person records for the united
  states, May 2017.

\bibitem{ShenA13}
Entong Shen and Lyublena Antova.
\newblock Reversing statistics for scalable test databases generation.
\newblock In {\em DBTest}, pages 7:1--7:6, 2013.

\bibitem{SnokeS18}
Joshua Snoke and Aleksandra~B. Slavkovic.
\newblock pmse mechanism: Differentially private synthetic data with maximal
  distributional similarity.
\newblock In {\em PSD}, volume 11126, pages 138--159, 2018.

\bibitem{SoltanaSB17}
Ghanem Soltana, Mehrdad Sabetzadeh, and Lionel~C. Briand.
\newblock Synthetic data generation for statistical testing.
\newblock In {\em ASE}, pages 872--882, 2017.

\bibitem{SoltanaSSB18}
Ghanem Soltana, Nicolas Sannier, Mehrdad Sabetzadeh, and Lionel~C. Briand.
\newblock Model-based simulation of legal policies: framework, tool support,
  and validation.
\newblock {\em Software and Systems Modeling}, 17(3):851--883, 2018.

\bibitem{Sweene02}
Latanya Sweeney.
\newblock k-anonymity: {A} model for protecting privacy.
\newblock {\em Int. J. Uncertain. Fuzziness Knowl. Based Syst.},
  10(5):557--570, 2002.

\bibitem{TayDWSLL13}
Y.~C. Tay, Bing~Tian Dai, Daniel~T. Wang, Eldora~Y. Sun, Yong Lin, and Yuting
  Lin.
\newblock Upsizer: Synthetically scaling an empirical relational database.
\newblock {\em Inf. Syst.}, 38(8):1168--1183, 2013.

\bibitem{tpch}
TPC.
\newblock Tpc-h benchmark.
\newblock \url{http://www.tpc.org/tpch/}, 2020.

\bibitem{williamson2002combinatorics}
Stanley~Gill Williamson.
\newblock {\em Combinatorics for computer science}.
\newblock Courier Corporation, 2002.

\bibitem{Winkler04a}
William~E. Winkler.
\newblock Masking and re-identification methods for public-use microdata:
  Overview and research problems.
\newblock In {\em PSD}, volume 3050, pages 231--246, 2004.

\bibitem{YahalomSZ10}
Ran Yahalom, Erez Shmueli, and Tomer Zrihen.
\newblock Constrained anonymization of production data: {A} constraint
  satisfaction problem approach.
\newblock In {\em Secure Data Management, 7th {VLDB} Workshop}, volume 6358,
  pages 41--53, 2010.

\bibitem{ZhangCPSX17}
Jun Zhang, Graham Cormode, Cecilia~M. Procopiuc, Divesh Srivastava, and Xiaokui
  Xiao.
\newblock Privbayes: Private data release via bayesian networks.
\newblock {\em {ACM} Trans. Database Syst.}, 42(4):25:1--25:41, 2017.

\end{thebibliography}
